\documentclass[11pt]{amsart}

\pdfoutput=1

\usepackage{amsmath, amsfonts, amssymb, color, enumitem}

\usepackage{comment}
\usepackage[utf8]{inputenc}
\usepackage{csquotes}
\usepackage[greek,english]{babel}
\usepackage[hyperref,doi,url=true,natbib,backend=biber,style=alphabetic,firstinits=true,sortcites=true,maxnames=5,sorting=nyt]{biblatex}
\RequirePackage[normalem]{ulem}
\usepackage[makeroom]{cancel}

\AtEveryBibitem{
  \clearlist{address}
  \clearfield{date}
  \clearfield{eprint}
  \clearfield{isbn}
  \clearfield{issn}
  \clearlist{location}
  \clearfield{month}
  \clearfield{series}
  \clearlist{language}
  \ifentrytype{unpublished}{}{
    \clearfield{url}
    \clearfield{urldate}}
  \ifentrytype{book}{}{
    \clearlist{publisher}
    \clearname{editor}
  }
}

\usepackage[bookmarks]{hyperref}
\usepackage{doi}

\usepackage{asymptote}

\begin{asydef}
texpreamble("\newcommand{\ol}[1]{\overline{#1}}\newcommand{\Rt}{\overline{R}}");
import fontsize;
defaultpen(fontsize(10pt));
\end{asydef}

\newtheorem{theorem}{Theorem}

\newtheorem{corollary}[theorem]{Corollary}
\newtheorem{definition}[theorem]{Definition}

\newtheorem{proposition}[theorem]{Proposition}
\newtheorem{remark}[theorem]{Remark}
\newtheorem{example}[theorem]{Example}

 \newcommand{\eps}{\varepsilon}

 \newcommand{\ntt}{\operatorname{ntt}}

\renewcommand{\epsilon}{\varepsilon}

\newcommand{\E}{\mathbb{E}}
\renewcommand{\P}{\mathbb{P}}
\newcommand{\N}{\mathbb{N}}
\newcommand{\Q}{\mathbb{Q}}
\newcommand{\R}{\mathbb{R}}
\newcommand{\W}{\mathbb{W}}

\newcommand{\ol}{\overline}

\newcommand{\NN}{\mathbb N}

\newcommand{\RR}{\mathbb R}

\newcommand{\Fcal}{\mathcal F}
\newcommand{\Gcal}{\mathcal G}
\newcommand{\QQ}{\mathbb Q}
\newcommand{\Acal}{\mathcal  A}
\newcommand{\Bcal}{\mathcal  B}
\newcommand{\Ccal}{\mathcal  C}

\newcommand{\Hcal}{\mathcal  H}

\newcommand{\Mcal}{\mathcal M}
\newcommand{\Scal}{\mathcal S}
\newcommand{\Pcal}{\mathcal P}

\newcommand{\Rcal}{\mathcal R}
\newcommand{\Xcal}{\mathcal X}

\newcommand{\lc}{[\![}
\newcommand{\rc}{]\!]}

\newcommand{\1}{\boldsymbol{1}}
\newcommand{\ie}{i.e.}

\newcommand{\eg}{e.g.}
\newcommand{\Omqv}{\Omega^{\mathsf{qv}}}
\newcommand{\wlg}{without loss of generality}
\newcommand{\wrt}{with respect to}
\newcommand{\usc}{upper semi-continuous}
\newcommand{\lsc}{lower semi-continuous}

\numberwithin{equation}{section}
\numberwithin{theorem}{section}

\newcommand{\di}{\mathrm{d}}

\newcommand{\SG}{\mathsf{SG}}
\newcommand{\RST}{\mathsf{RST}}

\begin{document}

\title[Model-independent pricing with insider information]{Model-independent pricing with insider information: a Skorokhod embedding approach}
\author{Beatrice Acciaio}
\address{Department of Statistics, The London School of Economics and Political Science, U.K.}
\email{b.acciaio@lse.ac.uk}
\author{Alexander M.~G.~Cox}
\address{Department of Mathematical Sciences, University of Bath, U.K.}
\email{a.m.g.cox@bath.ac.uk}
\author{Martin Huesmann}
\address{Institute for Mathematical Stochastics, Universit\"at M\"unster, Germany}
\email{martin.huesmann@uni-muenster.de}

\thanks{The authors are grateful for the excellent hospitality of the Hausdorff Research Institute for Mathematics (HIM), where the work was initiated. M.H. gratefully acknowledges  support of the  CRC 1060 and of the FWF-grant Y00782. During parts of this project M.H. has been funded by the Vienna Science and Technology Fund (WWTF) through project VRG17-005. In the final phase of the project M.H. has been funded by the Deutsche Forschungsgemeinschaft (DFG, German Research Foundation) under Germany's Excellence Strategy EXC 2044 –390685587, Mathematics Münster: Dynamics–Geometry–Structure.}  

\date{\today}

\begin{abstract}
In this paper we consider the pricing and hedging of financial derivatives in a model-independent setting, for a trader with additional information, or beliefs, on the evolution of asset prices. In particular, we suppose that the trader wants to act in a way which is independent of any modelling assumptions, but that she observes market information in the form of the prices of vanilla call options on the asset. We also assume that both the payoff of the derivative, and the insider's information or beliefs, which take the form of a set of impossible paths, are time-invariant.
In this way we accommodate drawdown constraints, as well as information/beliefs on quadratic variation or on the levels hit by asset prices.
Our setup allows us to adapt recent work of \citet{BCH17} to prove duality results and a monotonicity principle. This enables us to determine geometric properties of the optimal models.  Moreover, for specific types of information, we provide simple conditions for the existence of consistent models for the informed agent. Finally, we provide an example where our framework allows us to compute the impact of the information on the agent's pricing bounds.

\end{abstract}
\maketitle

\section{Introduction}

It has long been recognised that information plays an extremely important role in the study of modern financial markets. This is most markedly true when two parties trading the same asset have access to different information sources, and then one can ask how the `insider', who possesses additional information, should modify her behaviour to exploit her privileged position. In this paper, we aim to consider problems where the insider has a strong belief in some quantitative, or qualitative, fact about the future evolution of some asset, but is otherwise agnostic about other statistical properties determining the evolution of the asset.

A fundamental, motivating example will be the following: imagine an agent believes that the CEO of a company will act in such a way as to ensure that the share price of the company will not drop below a certain level which depends on the historical maximum of the share price, for example, because the manager is incentivised by stock options which payout provided this drawdown criteria is not breached. Then the agent may want to build this information into her valuations of e.g. derivatives written on the asset. The aim in this paper is to consider problems of this form in a model-independent framework. We claim that this is a natural framework for these problems, since the insider's information already rules out many `standard' models which would not usually satisfy such a constraint, and it is not immediately clear how the agent should choose a model which includes this information.

Problems concerning insider information have a rich literature: the first work in the mathematical finance literature is \citet{PK96}, while important subsequent work includes \cite{BO05,AIS98,GP98,C05}, and this topic is still a very active area of research. Along with different information sets, agents may have different beliefs on the evolution of asset prices. This again will result in different market behaviours.

In the past few years, robust approaches to finance, where no underlying probability measure is assumed \emph{a priori}, have become very popular. Only very recently, additional information and beliefs have been considered in a robust framework. In both \citet{AL15} and \citet{AHO16}, this has been modelled by an enlargement of filtration.  
Closer to the approach of the current work, are the papers by \citet{CHO14}, \citet{HO15}, \citet{BaKuNe17}, and \citet{BZZ18}, which model beliefs in a robust setting by excluding some paths from the possible evolution of the asset's price process (see Section~\ref{sect.lit}).

The goal of this paper is to consider the pricing and hedging problems for traders with different information and beliefs in a continuous-time, robust setting, where call prices at a fixed maturity $T$ are observed. Our analysis relies on two key assumptions. First, we only consider derivatives which are time-invariant, that is, with payoffs which are independent of the clock under which the underlying is running. These include, for example, lookback options, barrier options, corridor options, and variance options. Secondly, we assume that beliefs and insider's additional information are time-invariant and such that they allow the insider to assume that a certain set of paths is impossible. This means specifying the set of feasible paths on which (super-)hedging arguments are required to work. Examples of beliefs we can deal with include those on quadratic variation and those on asset prices hitting (or not hitting) given barriers. As for the time-invariant information, the main example we have in mind is that of drawdown constraints, e.g. imposed by company policy, on the price process never falling below a fixed fraction of its maximum-to-date or never falling below a certain threshold once it has reached a certain level.

The assumption of time invariance allows us to translate the robust pricing problem into a constrained Skorokhod embedding problem (SEP), emulating the approach to robust finance initiated in \citet{Ho98} (see also \cite{Ho11}). In this way,  in the first part of the paper, Sections \ref{sec:informed pricing} and \ref{sec:theo results}, we develop a theoretical framework for our approach.  We are able to extend to the current framework the analysis developed in \citet{BCH17} for the unconstrained problem. Indeed, a simple application of the results of \cite{BCH17} leads to a superhedging and duality result for the insider/the constrained SEP (Theorem~\ref{thm:superhedge}), and to a monotonicity principle which gives a necessary condition on the optimising probability measures for the insider/the constrained SEP (Theorem~\ref{thm:mp}).  Leveraging on our duality result, we are able to provide simple necessary and sufficient conditions to exclude arbitrage for the insider in terms of solutions to the constrained SEP (Proposition~\ref{prop:NA}). On the other hand, the monotonicity principle, that takes the form of geometric conditions on the support of the optimisers, often leads to barrier type solutions. Experience in the case without information suggests that, once the geometric structure of the support of the optimisers is understood, it is much simpler to e.g. develop numerical methods to compute the optimisers for  specific examples. The main motivation for considering the problem under this assumption of time-invariance is that, as a consequence, we are able to prove the monotonicity criteria, which, as we demonstrate, in many natural examples allows us to reformulate the optimisation problems in terms of simpler, geometric criteria. This additional insight opens up a wider range of mathematical and numerical tools which may be applied to both increase understanding of the problems, and to more accurately solve the problems.

In the second part, Section~\ref{sec:examples}, we illustrate that our setup allows to further investigate several concrete and financially meaningful situations gaining new insights in the insider's behaviour in explicit situations.
More precisely, we restrict ourselves to specific sets of feasible paths  
%which are either before or after a given stopping time or in between two stopping times 
(cf.\ \eqref{ex:l}). This class is quite broad, as it includes information and beliefs mentioned above on whether prices hit certain barriers, on whether the quadratic variation reaches certain levels, and on drawdown constraints. Here we are interested in three interrelated questions:
\begin{enumerate}
 \item When does there exist arbitrage for the insider?
 \item What are the worst case models for the insider?
 \item Can we calculate the value of the insider's information in specific situations?
\end{enumerate}
We address the question of arbitrage in Theorem \ref{thm:RSTExists} for the specific information encoded by \eqref{ex:l}. Specialising to concrete examples, we show in Theorems \ref{thm:root} and \ref{thm:exay} that the question of arbitrage can be reduced to simple ordering properties of particular functions. To the best of our knowledge, the present work is the first one to address the issue of arbitrage in a robust setting with additional information/beliefs.

Concerning the characterisation of worst case models we exemplify the power of the monotonicity principle in Theorem \ref{thm:constRoot} in a concrete setup. We consider the example of variance options with drawdown constraints, and show that we are able to 
\begin{enumerate}
\item Determine when there exists an arbitrage for the insider;
\item Characterise the class of extremal models;
\item Compute numerically the value of the insider's information.
\end{enumerate}
Specifically, in Section~\ref{sec:numerics} we give a numerical example to show the impact of increasing information on the insider's extremal model. Thereby, we can nicely illustrate the impact of increasing information; see Figure~\ref{fig:price}.

\subsection{Literature}\label{sect.lit}

In the robust approach to mathematical finance, the usual setting consists in having some assets available for dynamic trading, and some claims which are available at time zero for static, i.e.\ buy-and-hold, trading. The information at the disposal of the agent is the price of assets and claims at time zero, and the evolutions of the  price of assets in time. In this framework, most of the literature so far has been devoted to showing pricing-hedging duality results, that is, that the minimal cost to super-hedge pathwise a given derivative, equals its maximal price over calibrated martingale measures; see e.g.~\cite{BHP13,ABPS13,Nu14,BN15,BZ14,FH14} in discrete time, and \cite{GHT14,DS14a,DS14b,BCH17,BBKN14,BNT15,HO15,GTT15,BCHPP15} in continuous time, among a rapidly growing literature.

The current literature on the insider problem in a robust setup is still in its infancy.  
In \cite{AL15} and \cite{AHO16} the informed agent has a richer information flow, which results in having more choices for trading strategies, and hence in cheaper robust (super-hedging) prices.  In \cite{AL15} the authors study the models under which the market is complete in a semi-static sense, and through these models they compare the robust prices of the agents with and without additional information. In \cite{AHO16}, pricing-hedging duality results are  given when the additional information is disclosed either at time zero or at a given future instant in time, and it is given by specific random variables. 

Mathematically, our approach is closer to  \cite{CHO14}, \cite{HO15}, \cite{BaKuNe17}, \cite{BZZ18}. Those papers do not consider insider information, but model beliefs or prediction sets in a robust setting by specifying the set of feasible paths for the possible evolution of the asset's price process. The main aim in these papers is the study of the pricing-hedging duality.
In \cite{CHO14} the authors work in discrete-time and study duality, showing that in some cases a gap may appear, i.e. duality may fail. In \cite{HO15} a continuous-time setup is considered, and sufficient criteria are given so that duality holds in an asymptotic sense. In \cite{BaKuNe17} the authors consider a continuous-time setting and prediction sets in the space of continuous paths, and provide several duality results. Finally, in \cite{BZZ18} the authors obtain duality and monotonicity results for a broad class of constrained optimal transport problems, under some conditions on the space of paths and on the set of admissible transports. 

In the present paper we work in continuous time, so our duality results are comparable to those in \cite{HO15},\cite{BaKuNe17}, \cite{BZZ18}. However, \cite{HO15} considers derivatives with uniformly continuous payoff, so that the framework is orthogonal to the present one, where payoffs are assumed to be invariant to time change.  In \cite{BaKuNe17} these restrictions are substantially weakened, but without the inclusion of other traded options. In \cite{BZZ18} analyticity conditions are required on the set of admissible paths, rather than the time-invariance assumed here.
Also in a similar spirit to our results  in Section~\ref{sect:var} is the PhD thesis of \citet{Spoida}, which considers the situation where only finitely many options are available for static trading and, for specific kinds of derivatives, describes the optimal solutions for agents having beliefs on realised variance.

 To the best of our knowledge, the {\it constrained} Skorokhod embedding problem (conSEP) has not previously been systematically considered in the literature.
The only papers which we are aware of, that consider related problems, are \citet{AnSt11} and \citet{AnHoSt15}, that provide conditions under which a distribution may be embedded in Brownian motion or a diffusion in bounded time, which have some connections to the results in Section~\ref{sec:examples}. Also, the setting in \cite{BZZ18} covers for example the case of robust pricing in case of bounded quadratic variation, which leads to establishing conditions for the existence of Skorokhod embeddings in bounded time.

\subsection{Outline of the article}

In the present paper we will work in a continuous-time setup, under the assumption that the asset's price process $S$ evolves continuously, and all call options for a given maturity $T$ are traded at time zero in the market. We perform a time-change to formulate the pricing problem as a constrained optimal stopping problem in Wiener space and resort to Skorokhod embedding techniques. For this approach to be effective, we need to restrict our attention to the case where both the derivatives' payoff function and the feasibility of paths are invariant to time-changes in an appropriate sense.
The key concepts and definitions for this setup are introduced in Section~\ref{sec:informed pricing}. In Section~\ref{sec:theo results}, we show that the pricing-hedging duality and the monotonicity principle of \cite{BCH17}
can be extended in a natural way to our setting, thus allowing us to give a geometric characterisation of the support of the optimisers in the primal problem.

Then, in Section~\ref{sec:examples}, we consider specific examples of feasible sets where we may apply the results of the previous sections to determine specific consequences of certain types of information possessed by the insider. We first consider the implications of information which restricts the observed paths to occur either before or after (or both) some path-dependent event. In this case, we are able to give sufficient conditions for the existence of arbitrage for the insider. Next, we consider the case where feasibility corresponds to paths which do not enter given regions of an appropriate phase-space, and determine necessary and sufficient conditions for the additional information not to introduce arbitrage possibilities.  Finally,  we show how the monotonicity principle can be used to derive characterising properties of the optimisers subject to a given information set. In particular, we consider the problem described at the start of the introduction, where the insider believes a certain drawdown constraint is satisfied, and wishes to understand the impact on variance derivatives. In this case, we are able to describe the properties of the resulting optimiser, and also compute numerically an upper bound on the value of the derivative in both the model without information, and the model with information. These numerical results give us a good indication of the impact of the information on the prices of other derivatives.

\section{Informed robust pricing}\label{sec:informed pricing}
Throughout the paper, for $I \subset \R$, we write $\Ccal(I)$ for the space of continuous functions $\omega:I\to\R$ endowed with the topology of uniform convergence on compacts. When $I \subset [0,\infty)$, we write $\Ccal_x(I)$ for the subset of paths such that $\omega(0)=x$.

We consider a market consisting of a risk-free asset (bond), whose price is normalised to $1$, and a risky asset (stock) which is assumed to have a continuous price evolution, though neither  a reference probability nor the dynamics are specified. The assets are continuously traded on the fixed time-horizon $[0,T]$, $0<T<\infty$. Let the initial price of the stock be $s_0$; in this way we can think of the stock price process $S$ as the canonical process on $\Ccal_{s_0}[0,T]=\Ccal_{s_0}([0,T])$. We assume we observe the prices of call
options with maturity $T$ for all strikes, which corresponds to having the knowledge of the marginal distribution of $S$ at time $T$, say $\mu$, under any pricing measure by the Breeden-Litzenberger formula, \cite{BrLi78}. In particular, $\int x\mu(dx)=s_0$. We assume $\int (x-s_0)^2\,\mu(dx)=:V<\infty$. This condition is introduced in order to simplify the presentation, and can be relaxed (see e.g.\ \cite[Section 7]{BCH17}). Given a derivative with payoff function $F:\Ccal_{s_0}[0,T]\to\RR$ written on $S$, the robust pricing problem is to determine
\begin{equation}\label{pb:p}
\sup \{ \E_\Q[F(S)] : \Q\in\Mcal(\mu) \},
\end{equation}
where $\Mcal(\mu)$ is the set of all martingale measures $\Q$ on $\Ccal_{s_0}[0,T]$ such that $S_T\sim_\Q\mu$ (by martingale measure we mean a measure under which the canonical process is a martingale). 
This leads to the upper price bound for the derivative $F$ related to the worst case scenario of the evolution of the risky asset. Analogously, one can consider the infimum in \eqref{pb:p}, that is, the lower price bound for $F$. Mathematically the maximisation and minimisation problems are very similar, and in this article we concentrate on the former.

In practice, often not only the prices of call options with given maturity are available, but an agent may have other information or beliefs relating to the evolution of the asset price.
Incorporating this information may rule out certain behaviour of the stock price $S$, and hence certain models for $S$, which in turn leads to potentially smaller price bounds. We model this by introducing an \emph{informed agent}, also called the \emph{insider}, possessing some additional information  and beliefs which enables her to only consider a subset $\Acal\subseteq\Ccal_{s_0}[0,T]$ of \emph{feasible paths} for $S$ (precise assumptions on $\Acal$ will be given in \eqref{eq:feasible sets}). All other paths in $\Ccal_{s_0}[0,T]\setminus\Acal$ are deemed negligible due to the additional information held by the insider. Hence the robust pricing problem for the insider is
\begin{equation}\label{eq:insider}
P_\Acal:=\sup \{ \E_\Q[F(S)] : \Q\in\Mcal(\mu),\, \Q(\Acal)=1 \}.
\end{equation}
To give a value to the additional information, we will talk of the \emph{uninformed agent} or \emph{outsider} when considering an agent who does not have any other information than the call prices. Hence, the outsider's pricing problem is the classical robust pricing problem in \eqref{pb:p}, which corresponds to setting $\Acal=\Ccal_{s_0}[0,T]$ in \eqref{eq:insider}.

In the rest of this section, we recall and adapt the setup and results from \cite{BCH17} and \cite{BCHPP15} relying on \cite{Vo15}, which will allow us to formulate and analyze \eqref{eq:insider} as a \emph{constrained Skorokhod embedding problem}. In order to do so, we will first introduce a time-change, in Section~\ref{sect:tt}, that is a different clock under which we want to observe the paths of $S$. Next we will show that the pricing problem \eqref{eq:insider} has an equivalent formulation as an optimal stopping problem for Brownian motion on some probability space (problem \eqref{eq:P*f}), when the derivative and the additional information are invariant with respect to this time-change. Finally, we shall pass from this weak formulation of the problem to an optimisation problem on a single probability space, the Wiener space, which will require more general stopping rules; see problem \eqref{eq:Wiener opt}.

\subsection{Time transformation.}\label{sect:tt} 

The key tool to translate \eqref{eq:insider} into a constrained Skorokhod embedding problem is the Dambis-Dubins-Schwarz Theorem. However, we need to be careful in defining the time change since we want to be able to shift pathwise inequalities from $\Ccal_{s_0}[0,T]$ to the Wiener space and back. Moreover, the time change will be a useful tool to precisely define the options we want to consider as well as the set of feasible paths for the insider.

For $\omega\in\Ccal(\R_+)$ and $n\in\N$, we define the sequence of times
\[\sigma^n_0(\omega):=0,\quad \sigma^n_{k+1}(\omega):=\inf\{t>\sigma^n_k(\omega) : |\omega(t)-\omega(\sigma^n_k)|\geq2^{-n}\},\; k\in\N,\]
and we say that $\omega$ has quadratic variation if the sequence $(V_n(\omega))_{n\in\N}$ of functions
\[
V_n(\omega)(t):=\sum_{k=0}^\infty(\omega(\sigma^n_{k+1}\wedge t)-\omega(\sigma^n_k\wedge t))^2,\qquad t\in\R_+
\]
converges uniformly on compacts to some function in $\Ccal_0(\R_+)$, and the limit function has the same intervals of constancy as $\omega$.
We denote this function by $\langle \omega\rangle$. We write $\Omqv$ for the space of all paths $\omega$ in $\Ccal_{s_0}(\R_+)$ possessing such a quadratic variation and such that either $\langle \omega\rangle$ diverges at infinity or $\langle \omega \rangle$ is bounded and $\omega$ has a well defined limit at infinity. These conditions are necessary in order for the map $\ntt$ given below to be well defined. It is not hard to show that $\Omqv$ is a measurable subset of $\Ccal(\R_+)$.

We define the space of \emph{stopped paths} as
\[
\Scal:=\left\{(f,s) : f\in\Ccal_{s_0}[0,s], s\in\R_+\right\},
\]
and equip it with the distance $d_\Scal$ defined for $s<t$ by
\begin{equation*}
d_\Scal((f,s),(g,t))=\max\Big\{t-s,\sup_{0\leq u\leq s}|f(u)-g(u)|,\sup_{s\leq u\leq t}|g(u)-f(s)|\Big\},
\end{equation*}
which turns $\Scal$ into a Polish space.
The space $\Scal$ is a convenient way of encoding optionality of a process in our pathwise setup,  see e.g.\ \cite[Theorem IV.\ 97]{DeMeA};
note that optionality is equivalent to predictability, since we consider only continuous paths. More precisely, we set
$$ r: \Ccal_{s_0}(\R_+)\times \R_+ \to \Scal, \quad (\omega,t)\mapsto (\omega|_{[0,t]},t),$$
where $\omega|_{[0,t]}$ denotes the restriction of $\omega$ to $[0,t]$. Then a process $X$ with $X_0=s_0$ is optional if and only if there is a Borel function $H:\Scal\to\R$ such that $X=H\circ r$.

We call $\Omega_T^{\mathsf{qv}}$ the set of paths in $\Ccal_{s_0}[0,T]$ which have a continuation in $\Omega^{\mathsf{qv}}$, and for $\omega\in\Omega_T^{\mathsf{qv}}$ we define the following new clock:
\begin{equation}\label{eq:tau_t}
\tau_t(\omega)=\inf\{s\in[0,T] : \langle\omega\rangle_s > t\}\wedge T,\quad t\in\R_+,
\end{equation}
with the usual convention $\inf \emptyset=+\infty$.

We will work with the \emph{normalising time transformation} introduced by Vovk~\cite{Vo12}, which is defined by 
$\ntt_T:\Omqv_T\to\Scal$ given by
\begin{equation*}
\ntt_T(\omega)=((\omega_{\tau_t})_{t \le \langle\omega\rangle_T},\langle\omega\rangle_T).
\end{equation*}
That is, $\ntt_T(\omega)$ is a version of the
path $\omega$ run at a speed such that, for every $t$, its pathwise quadratic variation at time $t$ is exactly $t$. It will also be notationally useful at times to `forget' the time component, and consider the function $\ntt(\omega)$, which is equal to $(\omega_{\tau_t})_{t \le \langle\omega\rangle_T} \in \Ccal_{s_0}[0, \langle\omega\rangle_T]$. Of course, the two quantities are mathematically equivalent. The normalising time transformation will be the tool that will allow us to define the class of time-invariant derivatives and the kind of time-invariant additional information which are suitable in order to develop the SEP approach to  robust pricing with insider information.

\begin{remark}\label{rem:omqv}
 Note that $\QQ(\Omqv_T)=1$ for each $\QQ\in\Mcal(\mu)$ (see \citet{Ka95,Vo12}). For this reason, when studying the pricing problem \eqref{eq:insider}, we only consider paths in this set. 
\end{remark}

In this article, we consider payoff functions $F:\Ccal_{s_0}[0,T]\to\R$ which on $\Omega^{\mathsf{qv}}_T$ satisfy
\begin{equation}\label{eq:F}
F=\gamma\circ\ntt_T,
\end{equation}
for some Borel measurable $\gamma:\Scal\to\R$.
This means that the payoff function $F$ is identical for all paths which are time-transformations of each other, that is, which coincide after normalising the speed at which they run. 

A key additional component in our model will be the information which is held by the insider, and which is not known to the market. We will model this by assuming that the insider knows a set of feasible paths $\Acal\subseteq\Ccal_{s_0}[0,T]$. Thanks to Remark~\ref{rem:omqv}, we may assume \wlg{} that $\Acal\subseteq\Omqv_T$.
As with the payoff function, we will assume that the set $\Acal$ of feasible paths is time-invariant. More precisely, we will consider sets $\Acal$ given by
\begin{align}\label{eq:feasible sets}
\mathcal A = \ntt_T^{-1}(\Lambda),
\end{align}
for some measurable subset $\Lambda\subseteq \Scal$, so that $1_{\mathcal A}(\omega)=1_\Lambda \circ \ntt_T(\omega)$.
We will call $\Lambda$ the \emph{feasibility set}. In this way, feasibility of a path $\omega\in\Ccal_{s_0}[0,T]$ is shifted to admissibility of  the stopped path $(\ntt(\omega),\langle\omega\rangle_T)$. 
In particular, if a path $\omega \in \Ccal_{s_0}[0,T]$ is feasible, so is any other path which is a time transformation of $\omega$.

\subsection{Informed robust pricing as constrained SEP} The time transformation introduced above enables us to express the robust pricing problem \eqref{eq:insider} as a constrained optimal stopping problem for Brownian motion. 

\begin{proposition}\label{prop:conSEP}
Let $F$ and $\Acal$ satisfy \eqref{eq:F} and \eqref{eq:feasible sets}. The pricing problem for the insider
\eqref{eq:insider} can be formulated as 
\begin{equation}\label{eq:P*f}
P^*_\Lambda:=\sup\left\{\E[\gamma((W_t)_{t\leq\tau},\tau)] : \begin{aligned}
&(\tilde\Omega,(\Gcal_t)_{t\geq 0},\Gcal,\P) \textrm{supporting Brownian motion }W, \\ & W_0=s_0, \tau \textrm{ a } \Gcal\textrm{-stopping time s.t. } W_\tau\sim\mu,\\
& \, (W_{t\wedge\tau})_{t\geq 0}\ \textrm{is u.i.},\, \text{and}\ \E[1_\Lambda((W_t)_{t\leq\tau},\tau)]=1
\end{aligned}
\right\}.
\end{equation}
\end{proposition}

The condition $\E[1_\Lambda((W_t)_{t\leq\tau},\tau)]=1$ means that, when moving along a path of $W$, we can stop only at times such that the stopped path lies in $\Lambda$. This corresponds to the fact that informed agents only need to take into account the paths in the feasibility set, $\Lambda$.
The condition of uniform integrability on $W_{\cdot\wedge\tau}$
is - in the current setup - equivalent to $\tau$ being minimal, cf.\ \cite{Monroe:1972aa}. This means that, for any other stopping time $\tau'$ in the same filtered probability space,
\begin{align}\label{eq:minimal}
\tau'\leq \tau\, \text{ and }\, W_{\tau'}\sim \mu\, \text{ imply }\, \tau'=\tau\, \text{a.s.}
\end{align}

The reformulation in \eqref{eq:P*f} shows how the maturity time $T$, at which we know the marginal distribution $\mu$ of the price process, does not play any role in the pricing problem. This is a consequence of the time-invariance assumption.

\begin{proof}
The proof of this essentially follows from \cite[Section~4]{BCHPP15}:
 Let $\Q\in\Mcal(\mu)$, $(\tau_t)_{t\in\R_+}$ be the time-change defined in \eqref{eq:tau_t}, and $(\Fcal^S_t)_{t\in[0,T]}$ 
the usual augmentation of the filtration generated by $(S_t)_{t\in[0,T]}$. It is easy to verify that
$\langle S\rangle_T$ is a stopping time with
respect to the filtration $(\Fcal^S_{\tau_t\wedge T})_{t\in\R_+}$. Then, the Dambis-Dubins-Schwarz theorem implies that the
process $(X_t)_{t\in\R_+}=(\ntt(S)_{t\wedge\langle S\rangle_T})_{t\in\R_+}$ is a
stopped Brownian motion under $\Q$ in the filtration
$(\Fcal^S_{\tau_t\wedge T})_{t\in\R_+}$. Moreover, it is uniformly integrable
and satisfies $\ntt(S)_{\langle S\rangle_T}\sim\mu$.  \emph{Vice versa}, let $W$ be a Brownian motion on some probability space $(\tilde\Omega,(\Gcal_t)_{t\geq0},\P)$, and
$\tau$ be a stopping time such that $W_{.\wedge\tau}$ is uniformly
integrable with $W_\tau\sim\mu$. Then, for $M=(M_t)_{t\in[0,T]}$
defined by $M_t:=W_{\frac{t}{T-t}\wedge\tau}$, we have 
that
$\P\circ
M^{-1}\in\Mcal(\mu)$. 
The result follows.
\end{proof}

To be able to analyse the optimisation problem \eqref{eq:P*f}, we introduce another optimisation problem living on a single probability space, the Wiener space $(\Ccal_{s_0}(\R_+),\Fcal,\W)$. To this end we consider the set
\[
M=\{\xi\in\Pcal(\Ccal_{s_0}(\R_+)\times\R_+) : \xi(d\omega,dt)=\xi_\omega(dt)\W(d\omega), \xi_\omega\in\Pcal(\R_+)\, \textrm{for $\W$-a.e. $\omega$}\},
\]
where  $\Pcal(\Xcal)$ denotes the set of probability measures on a space $\Xcal$, and $(\xi_\omega)_{\omega\in\Ccal_{s_0}(\R_+)}$ is a  regular disintegration of $\xi$ with respect to the first coordinate $\omega$. We equip $M$ with the  weak topology induced by the continuous bounded functions on $\Ccal_{s_0}(\R_+)\times\R_+$. Each $\xi\in M$ can be uniquely characterised by the cumulative distribution function $A^{\xi}(\omega,t)=\xi_\omega[0,t]$.

\begin{definition}\label{def:RST}
We say that a measure $\xi\in M$ is a \emph{randomised stopping time} if the corresponding increasing process $A^\xi$ is optional, and write $\xi \in \RST$. For an optional process $X:\Ccal_{s_0}(\R_+)\times \R_+\to\R_+$ and $\xi\in\RST$, we define $X_\xi$ as the pushforward of $\xi$ under the mapping $(\omega,t)\mapsto X_t(\omega)$.
We denote by $\RST(\mu)$ the set of all randomised stopping times such that $W_\xi=\mu$ and $\int t~\xi(d\omega, dt)<\infty$. 
\end{definition}

\begin{remark}\label{rem:RST}
It is well known that any randomised stopping time $\xi$ can be identified with a stopping time $\tau_\xi$ on the extended probability space $(\Ccal_{s_0}(\R_+)\times [0,1], \Fcal\otimes \mathcal B,\W\otimes\mathcal L)$, where $\mathcal B$ denotes the Borel $\sigma$-algebra on $[0,1]$, and $\mathcal L$ the Lebesgue measure on $[0,1]$. One way of defining $\tau_\xi$ is via
$$\tau_\xi(\omega,u):=\inf\{t\geq 0: \xi_\omega([0,t])\geq u\}.$$
As a consequence, the optional stopping theorem applies for randomised stopping times. Indeed, any process $X$ on $\Ccal_{s_0}(\R_+)$ can be lifted to a process $\bar X$ on $\Ccal_{s_0}(\R_+)\times [0,1]$ by setting $\bar X_t(\omega,u):= X_t(\omega)$ and then the classical optional stopping theorem applies for, e.g., uniformly integrable martingales.
\end{remark}

Considering the martingale $W_t^2-t$ it follows from classical results on stopping times (e.g.\ \cite[Corollary 3.3]{Ho11}, \cite[Lemma 3.12]{BCH17}, Remark \ref{rem:RST}) that, for $\xi\in\RST$ with $W_\xi=\mu$, the condition $\int t~\xi(d\omega, dt)<\infty$ is equivalent to 
\begin{align}\label{eq:V}
\int t~\xi(d\omega, dt)=\int (x-s_0)^2 ~\mu(dx) = V, 
\end{align}
which is assumed to be finite in our setup.

By \cite[Theorem 3.14]{BCH17}, $\RST(\mu)$ is non-empty and compact \wrt{} the topology induced by the continuous and bounded functions on $\Ccal_{s_0}(\R_+)\times \R_+$. As a direct consequence we get the following result:
\begin{corollary}\label{cor:compact}
Let $\Lambda\subseteq\Scal$ be closed. Then the set of feasible randomised stopping times 
\begin{equation}\label{eq:RSTf}
\mathsf{RST}(\mu; \Lambda):=\left\{\xi\in\mathsf{RST}(\mu) : \int_{\Ccal_{s_0}(\R_+)\times\R_+}
1_\Lambda\circ r(\omega,t)\xi(d\omega,dt)=1\right\}
\end{equation}
is convex and compact \wrt{} the topology induced by the continuous and bounded functions on $\Ccal_{s_0}(\R_+)\times \R_+$.
\end{corollary}
We highlight here the important feature that $\RST(\mu;\Lambda)$ might be empty, which can be understood as a robust arbitrage opportunity, see Proposition \ref{prop:NA} and Section \ref{sec:examples}.
\begin{proof}
 Since $\Lambda$ is assumed to be closed, the function $1_\Lambda\circ r$ is u.s.c. Hence $\int_{\Ccal_{s_0}(\R_+)\times\R_+}
1_\Lambda\circ r(\omega,t)\xi(d\omega,dt)=1$ is a closed condition by the Portmanteau theorem. 
\end{proof}

Another important property of the feasible randomised stopping times is that they are precisely the joint distributions on $\Ccal_{s_0}(\R_+)\times\R_+$ of pairs $(W,\tau)$ satisfying the constraints in \eqref{eq:P*f}. This is a straightforward extension of \cite[Lemma 3.11]{BCH17}. Putting everything together we have derived a formulation of our optimisation problem \eqref{eq:insider} resp.\ \eqref{eq:P*f} on the Wiener space as a \textit{constrained Skorokhod embedding problem}:
\begin{proposition}\label{prop conSEP}
In the setting described above,
\begin{equation*}\label{eq:Wiener opt}
P^*_\Lambda=\sup\left\{\int_{\Ccal_{s_0}(\R_+)\times\R_+}\gamma\circ r(\omega,t)\xi(d\omega,dt) : \xi\in\mathsf{RST}(\mu;\Lambda)\right\}. \tag{$\mathsf{conSEP}$}
\end{equation*}
\end{proposition}

We will say that \eqref{eq:Wiener opt} is well posed if $\textstyle{\int_{\Ccal_{s_0}(\R_+)\times\R_+}\gamma\circ r~d\xi}$ exists with values in $[-\infty,\infty)$ for all $\xi \in\RST(\mu;\Lambda)$, and is finite for one such $\xi$. In particular, \eqref{eq:Wiener opt} is not well posed if $\RST(\mu;\Lambda)=\emptyset$ which has a pleasing financial interpretation (cf.\ Proposition~\ref{prop:NA}).
The (unconstrained) Skorokhod embedding problem corresponds to the case $\Lambda=\Scal$, when all paths are feasible, hence the above supremum is taken over $\mathsf{RST}(\mu)$.

From an analytical point of view, the formulation \eqref{eq:Wiener opt} is extremely useful since we are now dealing with a linear optimisation problem over a convex and compact set on a single probability space. A direct consequence is the following result:

\begin{theorem}\label{thm:max exists}
Let $\gamma:\Scal\to\R$ be \usc{} and bounded from above in the sense that, for some constants $a, b, c\in\R_+$,
 \begin{align}\label{eq:bdd above}
\gamma((\omega(s))_{s\leq t},t)\leq a+bt+c\sup_{s\leq t}\omega(s)^2,\qquad (\omega,s)\in\Ccal_{s_0}(\R_+)\times\R_+.
 \end{align}
Assume that {$\Lambda\subseteq\Scal$ is closed} and that $\RST(\mu;\Lambda)$ is non-empty. Then the optimisation problem \eqref{eq:Wiener opt} admits a maximiser.
\end{theorem}
\begin{proof}
We claim that \wlg{} we can assume that $\gamma$ is bounded from above. Indeed, by the pathwise version of Doob's inequality (see \cite{AB+13}),
$$ \sup_{s\leq t} \omega(s)^2\leq M_t + 4\omega(t)^2$$
for some martingale $M_t$ starting in zero. Hence condition \eqref{eq:bdd above} implies that 
$$\tilde\gamma\circ r(\omega,t):=\gamma\circ r(\omega,t) - a' - b't -c' (M_t + \omega(t)^2)$$
 is bounded from above and the term $\int a' - b't -c' (M_t + \omega(t)^2)~ d\xi$ is independent of $\xi\in \RST(\mu)$ by \eqref{eq:V} and the assumed second moment of $\mu$. Therefore, we can assume $\gamma$ to be bounded from above.

Finally, since $\RST(\mu;\Lambda)$ is compact  and $r$ continuous, and by the Portmanteau Theorem the map $\xi\mapsto \int \gamma\circ r~d\xi$ is \usc{}, we deduce the result.
\end{proof}

\section{Super-replication and monotonicity principle}\label{sec:theo results}

In this section, we show that a straightforward application of the results in \cite{BCH17} leads to duality or superhedging results, and to a geometric characterisation of primal optimisers, that is, to the monotonicity principle for constrained Skorokhod embedding.

\subsection{Duality}
In this section, we first show a duality result for the problem $P_\Lambda^*$ defined in \eqref{eq:P*f}, that is for \eqref{eq:Wiener opt}, and then from it deduce a duality result for the original robust pricing problem $P_\Acal$ defined in \eqref{eq:insider}. The latter is the analogue of the super-replication duality in the present robust setting with additional information/beliefs. As in the classical (non-robust) case, this in turn leads to a dichotomy between existence of martingale measures and existence of arbitrage opportunities, the so-called fundamental theorem of asset pricing, which we prove at the end of this section.

A martingale $\phi$ is called $\Scal$\!-continuous if there exists a continuous $H:\Scal\to\R$ such that $\phi=H\circ r$. Note that a martingale which is $\Scal$\!-continuous has continuous paths, but the other implication is in general not true.

\begin{theorem}\label{thm:constrSEPdual}
 Let $\gamma :\Scal \to \R$ be \usc{} and bounded from above  in the sense of \eqref{eq:bdd above}, and $\Lambda\subseteq\Scal$ be closed. Set
\begin{align}\label{eq:dualconstraint}
D^*_\Lambda:=\inf\left\{ \int\psi d\mu : \begin{array}{l}
\psi \in \Ccal(\R), \exists \mbox{ an $\Scal$-continuous martingale $\phi$, $\phi_0=0$ s.t.}\\ 
\phi_t(\omega)+\psi(\omega(t))\geq \gamma\circ r(\omega,t) \text{ for all } (\omega,t)\in r^{-1}(\Lambda)\end{array}\! \right\},
\end{align}
where $\phi, \psi$  satisfy 
\begin{equation}\label{eq pp bdd}
|\phi_t(\omega)| \leq a + bt+c \omega(t)^2,\;\; |\psi(y)| \leq a+ b y^2,\quad \forall\ \omega\in\Ccal_{s_0}(\R^+),\; \text{for some $a,b,c>0$}.
\end{equation}
Then we have
$$ P_\Lambda^*=D_\Lambda^*~.$$
\end{theorem}
\begin{proof}
Put 
\begin{equation}\label{eq.bargamma}
\overline\gamma(f,t)= 
\begin{cases}
  \gamma(f,t) & \quad (f,t) \in \Lambda,\\
  -\infty & \quad (f,t) \in\Scal \setminus \Lambda.
\end{cases}
\end{equation}
Since we assume that $\gamma$ is \usc{} and bounded from above in the sense of \eqref{eq:bdd above}, it is easy to see that these properties are inherited by $\overline\gamma$.
Moreover, the well-posedness assumption of \eqref{eq:Wiener opt} for $\gamma$ implies that \eqref{eq:Wiener opt} is still well-posed for $\overline\gamma$ and $\Lambda=\Scal$. Hence, the result follows from \cite[Theorem 4.2]{BCH17}.
\end{proof}

This duality result is already of interest in its own right. However, to identify it as a superreplication result we need to recover the hedging strategies corresponding to the martingale. For this we need some kind of pathwise martingale representation theorem. In fact Theorem 6.2 of \cite{Vo12} can be interpreted as such. To this end, we need to introduce some more notation.

We will need the concept of simple strategy, by which we mean a process $H:\Omega_T^{\mathsf{qv}}\times\RR^+\to\RR$ of the form
\[
H_t(\omega)=\sum_{n\geq 0}K_n(\omega)1_{(\tau_n(\omega),\tau_{n+1}(\omega)]}(t),\quad (\omega,t)\in \Omega_T^{\mathsf{qv}}\times\RR^+,
\]
where $0=\tau_0(\omega)<\tau_1(\omega)<\ldots$ are $\Fcal^S$-stopping times such that for every $\omega$ one has $\lim_{n\to\infty}\tau_n(\omega)=\infty$, and $K_n:\Omega_T^{\mathsf{qv}}\to\RR$ are $\Fcal^S_{\tau_n}$-measurable bounded functions for $n\in\NN$. For such a strategy, we can define the corresponding pathwise stochastic integral as 
\[
(H\cdot S)_t(\omega)=\sum_{n\geq 0}K_n(\omega)(S_{\tau_{n+1}(\omega)\wedge t}-S_{\tau_{n}(\omega)\wedge t})(\omega).
\]
Then, following exactly the line of reasoning as for the proof of \cite[Theorem 3.1]{BCHPP15} one can get the following result. 
We recall that $F = \gamma \circ \ntt_T$ and $\Acal = \ntt^{-1}_T (\Lambda)$,  and that the robust pricing problem for the insider was defined in \eqref{eq:insider}.

\begin{theorem}\label{thm:superhedge}
 Let $\gamma$ be \usc{} and bounded from above  in the sense of \eqref{eq:bdd above}, and let $\Lambda\subseteq\Scal$. Set
$$D_\Acal:=\inf\left\{\int \psi(y)\, d\mu(y): \begin{array}{l}
 \psi \in \Ccal(\R), \exists \mbox{ simple strategies}\ (H^n)_n\ \textrm{s.t.} \\
\liminf_n(H^n\cdot S)_T
(\omega)+\psi(\omega(T))\geq
F(\omega)\ \mbox{for all $\omega\in\Acal$}
 \end{array}\!
 \right\},
 $$
 where $|\psi(y)| \leq a+ b y^2$ and $(H^n\cdot S)_t\geq -a-bt$ for some $a,b>0$ and all $t \in [0,T]$. Then we have
$$P_\Acal=D_\Acal.$$
\end{theorem}

Theorem \ref{thm:superhedge} is the analogue of the classical super-replication duality theorem, in the present robust insider setting. Moreover, like its classical counterpart, it additionally implies a version of the first fundamental theorem of asset pricing. In the following we will use Theorem \ref{thm:superhedge} with different payoff functions. To stress the dependence on the cost function we will sometimes write $P_\Acal(F), D_\Acal(F)$.
\begin{proposition}\label{prop:NA}
 Under the assumptions of Theorem~\ref{thm:superhedge}, the following are equivalent:
\begin{itemize}
\item[(i)] $\exists~\QQ\in\Mcal(\mu)$ such that $\QQ(\Acal)=1$;
\item[(ii)] $\RST(\mu;\Lambda)\neq\emptyset$;
\item[(iii)] $\nexists~\epsilon>0$, simple strategies $(H^n)_n$, and $\psi\in\Ccal(\R)$ with $\int \psi~d\mu = 0$ such that 
   \begin{equation}
    \label{eq:NA1}
   \liminf (H^n\cdot S)_T(\omega) + \psi(\omega(T)) \ge \epsilon, \text{ for all }  \omega \in \Acal.
  \end{equation}  
\end{itemize}
\end{proposition}

Property $(iii)$ means that one cannot make arbitrary profits by starting with zero capital. Indeed, if \eqref{eq:NA1} holds for some $\hat\epsilon>0$, then it does so for any $\epsilon>0$.
\begin{proof}
The equivalence between $(i)$ and $(ii)$ follows from the arguments around Proposition \ref{prop:conSEP}.

$(i)\Rightarrow(iii)$:\ Note that $(i)$ implies $D_\Acal(\tilde F_0)=P_\Acal (\tilde F_0) =0$ for any derivative $\tilde F_0$ s.t. $\tilde F_0=0$ on $\Acal$, by Theorem~\ref{thm:superhedge}. Pick 
$$ F_0=\begin{cases}
        0 & \text{ on } \Acal\\
-\infty & \text{ else }
       \end{cases}.$$
Suppose, for contradiction, that there exist $\epsilon, (H^n)_n$ and $\psi$ s.t. \eqref{eq:NA1} is satisfied. Then, 
the  pair $((H^n)_n,\psi-\epsilon)$ is admissible for the dual problem $D_\Acal(F_0)$. However, this implies $ D_\Acal(F_0)\leq-\epsilon$ for $F_0$, which gives the desired contradiction.\\
$(iii)\Rightarrow(i)$:\
By Theorem~\ref{thm:superhedge}, if there is no measure $\QQ\in\Mcal(\mu)$ such that $\QQ(\Acal)=1$, then $D_\Acal=P_\Acal = -\infty$ for all derivatives $F$. In particular,  for 
$$ F_\eps=\begin{cases}
        \eps & \text{ on } \Acal\\
-\infty & \text{ else }
       \end{cases},$$
there exist $(H^n)_n$ and $\psi$, with $\int \psi~d\mu = 0$, such that \eqref{eq:NA1} holds.
\end{proof}

\begin{remark}
In this paper, we have only considered the case where the information of future call prices at a single fixed time $T$ is observed. Using similar methods to those developed in \cite{BCHPP15}, it is also possible to extend Theorems~\ref{thm:constrSEPdual} and \ref{thm:superhedge} to the case where the call prices at times
$0 \le s_1 \le s_2 \le \dots \le s_N = T$ are observed, and provide a related formulation in the Brownian setup where the optimisation is over a sequence of stopping times $\tau_1 \le \tau_2 \le \dots \le \tau_N = \tau$. In this case, it is possible to consider both the cases where call price information completely fixes the distributions at the intermediate times, or it only determines the integral of particular functions, or there is a mixture of some times having full information and others lacking it. In this more general setup, it becomes possible to include a large class of options, for example, a robust approach to discretely monitored Asian options could be included.
\end{remark}

\subsection{Constrained monotonicity principle}
In this section, we provide a modified version of the monotonicity principle of \cite{BCH17} giving necessary geometric conditions on the support set of an optimiser to \eqref{eq:Wiener opt}.

To this end, we denote the concatenation of two paths $(f,s),(g,t)\in\Scal$ by $f\oplus g$, i.e.
$$ f\oplus g(u)=\begin{cases}
                 f(u) & u \leq s,\\
f(s)+g(u-s)-g(0) & s\leq u\leq  s+t.
                \end{cases}$$
For $(f,s)\in\Scal$ we define the process $\gamma^{(f,s)\oplus}(\omega,t):= \gamma(f\oplus \omega|_{[0,t]},s+t).$

\begin{definition}\label{def:SGpairs}
 A pair $((f,s),(g,t))\in\Scal\times\Scal$ is called \emph{feasible stop-go pair}, written $((f,s),(g,t))\in\SG_\Lambda$, if $f(s)=g(t)$, $(f,s)\in\Lambda$,  the set of $(\Fcal_t^W)_{t\geq 0}$ stopping times $\sigma$ satisfying $0<\E[\sigma]<\infty$ and $1_\Lambda\circ r(f\oplus W,s+\sigma)=1$ a.s. is non-empty, and every such stopping time satisfies 
\begin{align}\label{eq:SGpairs}
 \E[\gamma^{(f,s)\oplus}((W_u)_{u\leq\sigma},\sigma)] + \gamma(g,t) < \gamma(f,s) + \E[\gamma^{(g,t)\oplus}((W_u)_{u\leq\sigma},\sigma)]~,
\end{align}
and  $1_\Lambda\circ r(g\oplus W,t+\sigma)=1$ a.s., where both sides of \eqref{eq:SGpairs} are well defined and the left hand side is finite. Here, the probability space is assumed to be rich enough to support a Brownian motion $W$, 
%and an uniform random variable independent of $W$, 
and $(\Fcal^W_t)_{t\geq 0}$ denotes the natural filtration generated by $W$.
\end{definition}
 The interpretation is that on average it is better to \emph{stop} a path at time $s$ with history $f$, and to run the paths that would have carried on from $(f,s)$ from a previously stopped history $(g,t)$ (to let $(g,t)$ \emph{go}),
as long as this results in a feasible stopping rule.  Note that since $f(s) = g(t)$, the law of the stopped process is not changed.
We remark here that -- as a consequence of only considering $(\Fcal^W_t)_{t\geq 0}$ stopping times -- the definition of feasible stop-go pairs is independent of the probability space on which $\sigma$ lives as long as it is rich enough to support the Brownian motion $W$.   
%and a uniform random variable independent of $W$.
% This is reminiscent of the fact that we have to consider randomised stopping times in  \eqref{eq:Wiener opt}.
 In a similar manner to \cite[Section 5]{BCH17}, one could introduce an even stronger notion of feasible stop-go pairs only considering one particular candidate stopping time.  In this article, we do not need this generality.

% \begin{remark}\label{rem:SG}
%  Writing $\SG=\SG_\Scal$ for the non-constrained stop-go pairs it immediately follows from the definition that $\SG\subset \SG_\Lambda$ for any $\Lambda\subset \Scal$ since pairs in $\SG_\Lambda$ have to satisfy less constraints.
% \end{remark}

For a set $\Gamma\subset \Scal$ we denote by $\Gamma^<$ the set of all stopped paths which have a proper extension in $\Gamma$:
\[
\Gamma^<:=\{(f,s)\in\Scal : \exists (g,t)\in\Gamma,\, s<t,\, g|_{[0,s]}=f\}.
\]
\begin{definition}\label{def:gamma monotone}
 A set $\Gamma\subset \Lambda$ is called \emph{feasible $\gamma$-monotone} if
$$\SG_\Lambda\cap(\Gamma^<\times\Gamma)=\emptyset~.$$
\end{definition}

A set $\Gamma\subset \Scal$ should be viewed as a possible stopping set, i.e.\ a set of paths $(\omega|_{[0,\tau]},\tau)$
for an admissible stopping strategy $\tau$ in \eqref{eq:Wiener opt}. If such a set $\Gamma$ is feasible $\gamma$-monotone then there is no way of changing the stopping rule in a pathwise fashion, as in \eqref{eq:SGpairs}, resulting in a feasible stopping rule with higher payoff.

\begin{theorem}[Constrained Monotonicity Principle]\label{thm:mp}
 Let $\gamma:\Scal\to\R$ be Borel. Assume that \eqref{eq:Wiener opt} is well posed and that $\xi\in\RST(\mu;\Lambda)$ is an optimiser. Then there exists a feasible $\gamma$-monotone set $\Gamma\subset\Scal$ such that 
$$\xi(r^{-1}(\Gamma))=1.$$
\end{theorem}
\begin{proof}
Taking $\overline\gamma$ as in \eqref{eq.bargamma}, the result follows from \cite[Theorem 5.7]{BCH17}.
\end{proof}

The Constrained Monotonicity Principle will be an important tool to characterise solutions to \eqref{eq:Wiener opt}, and in particular will allow us to deduce geometric features of optimisers. We will illustrate this in the subsequent sections.

\section{No-arbitrage, pricing and hedging in specific information settings}\label{sec:examples}

Up to now we have shown that under our assumptions the robust pricing problem \eqref{eq:insider}
can be reformulated as a constrained Skorokhod embedding problem for which we have established general results on existence, superheding, a variant of the first fundamental theorem of asset pricing, and a characterisation of optimisers. 

The goal of this section is to illustrate the richness of our framework by
considering some natural choices for the insider's information set $\Acal$, or equivalently for the corresponding feasibility set $\Lambda$, 
to show that under additional assumptions we are able to prove a variety of very explicit results in the insider's setting.

We use the notation $\prec$ to denote the convex order relation between probability measures; specifically, we say that $\lambda \prec \mu$ if $\int c(x) \lambda(dx) \le \int c(x) \mu(dx)$ for any convex function $c$.

In the examples we consider, we will typically address three related questions: 
\begin{enumerate}
\item Given a pair $(\mu, \Lambda)$, when does there exist any consistent model for the insider agent? Specifically, is $\RST(\mu;\Lambda)$ non-empty? We address these points in Theorems~\ref{thm:RSTExists}, \ref{thm:root}, and \ref{thm:exay}.

\item Assuming $\RST(\mu;\Lambda)\neq\emptyset$, can we characterise the worst case scenarios for the insider, i.e.\ can we characterise solutions to the constrained Skorokhod embedding problem?  We provide a characterisation of the optimisers to a specific problem in Theorem~\ref{thm:constRoot}.

\item Given a pair $(\mu, \Lambda)$ such that $\RST(\mu;\Lambda) \neq \emptyset$, and a derivative with payoff $F$, what is the value of $P_\Lambda^*$, and how does this differ from $P_\Scal^*$, the price of the uninformed agent? 
We answer these questions in the context of a specific example in Section~\ref{sect:var}.

\end{enumerate}

In investigating the questions above, we will focus on  the three following natural examples where the additional information/beliefs translates into stopping the Brownian motion after and/or before given stopping times.  Let
$\underline{\tau}, \overline\tau$ be  stopping times such that $\underline{\tau}\leq\overline\tau$, and $(W_{t\wedge\underline{\tau}})_{t\geq 0}, (W_{t\wedge\overline{\tau}})_{t\geq 0}$ are uniformly integrable, and consider the sets
\begin{equation}\label{ex:l}
\Lambda_1=\{r(\omega,t) : t\leq\overline{\tau}(\omega)\},\; \Lambda_2=\{r(\omega,t) : t\geq\underline\tau (\omega)\},\; \Lambda_3=\{r(\omega,t) : \underline{\tau}(\omega)\leq t\leq\overline\tau (\omega)\}.
\end{equation}
These cases notably cover the examples of additional information and beliefs mentioned at the beginning of the paper, whether prices hit certain barriers, whether the quadratic variation reaches certain levels (cf.\ Section \ref{sect:Root}), and on drawdown constraints (cf.\ also Section \ref{sect:AY}).

In this situation we have the following basic result on the existence or absence of a consistent model for the insider.

\begin{theorem}\label{thm:RSTExists}
Suppose that the insider has information given by \eqref{ex:l}. We write $W_{\underline{\tau}}\sim\underline{\mu}, W_{\overline\tau}\sim\overline\mu$. Then:
\begin{enumerate}[label=(\arabic*)]
  \item \label{item:1_RST} $\Lambda = \Lambda_1$: the set $\RST(\mu;\Lambda) = \emptyset$
    if $\mu \nprec \overline\mu$;
  \item \label{item:2_RST} $\Lambda = \Lambda_2$: the set $\RST(\mu;\Lambda) = \emptyset$
    if and only if $\underline\mu \nprec \mu$;
  \item \label{item:3_RST} $\Lambda = \Lambda_3$: the set $\RST(\mu;\Lambda) = \emptyset$
    if $\underline{\mu} \nprec \mu$ or $\mu \nprec \overline{\mu}$.
\end{enumerate}
In particular, if any of the conditions on the measures  $\underline{\mu}, \mu, \overline{\mu}$ above hold, then the insider can make unlimited profit in the sense of \eqref{eq:NA1}.
\end{theorem}

\begin{proof}
 As a consequence of Strassen's Theorem \cite{Strassen:1965aa},  a solution to the constrained problem \eqref{eq:P*f} exists for $\Lambda_2$ if and only if $\underline{\mu}\prec\mu$. Similarly, in the case of $\Lambda_1$, the condition $\mu \prec \overline{\mu}$ is a necessary condition for the existence of a stopping time $\tau \le \overline \tau$ for the Brownian motion such that $W_{\tau} \sim \mu$, but it is not sufficient unless $\overline \mu$ is supported on two points due to the result of Meilijson \cite{Me82} and van der Vecht \cite{vdV86}.\footnote{A simple example can be constructed by considering the measures $\overline{\mu} = N(0,1)$, with stopping time $\overline{\tau} = 1$ and $\mu = \frac{\eps}{2} (\delta_1 + \delta_{-1}) + (1-\eps) \delta_0$. For $\eps$ sufficiently small, it is easily checked that $\mu \prec \overline{\mu}$, but there is no bounded stopping time embedding $\mu$. } Combining these two observations yields the third item.
\end{proof}

 To the best of our knowledge, necessary and sufficient conditions for the existence of $\xi\in\RST(\mu;\Lambda_1)$ are unknown. We are able to provide them in specific settings (see Section~\ref{sec:barrier}), while the existence of general criteria remains an interesting open problem.

%We summarise these results in the following theorem.
Before we proceed, we would like to remark on the specific form of the feasibility sets $\Lambda$ in \eqref{ex:l}.
It is clear that, in general, not all information processes/feasibility sets are of the form \eqref{ex:l}. A full classification and analysis is beyond the scope of this paper. 
One of the various reasons that makes this analysis complicated is that the constraint $\tau \in \RST(\mu)$ may
impose additional conditions that are not immediate from the
construction of $\Lambda$. Consider for example the case where
\begin{equation}\label{eq:34}
\Lambda = \Big\{(f,t) \in \Scal : \sup_{s \le t} f(s)-c \le f(t)\Big\},
\end{equation}
for some fixed $c\in\R_+$, which corresponds to the drawdown constraint on the price process not dropping more than $c$ below its maximum-to-date value.
Minimality (cf.\ \eqref{eq:minimal}) implies that an admissible stopping time must occur before
$\overline{\tau} := \inf\{t \ge 0: \sup_{s \le t} \omega(s)-c > \omega(t)\}$,
by a simple martingale argument. Hence, although there exist feasible paths in $\Lambda$ which live longer
than $\overline{\tau}$, any $\tau$ which is in $\RST(\mu)$ must, with probability one, be bounded above by $\overline{\tau}$. Therefore, the set of feasible
stopped paths in this case  can be replaced by $\Lambda'=\{r(\omega,t) : t\leq\overline{\tau}\}$. Then, from the argument above, $\mu\prec\overline \mu\sim W_{\overline{\tau}}$ must hold
in order to have  a solution to the constrained embedding problem.  
\medskip\\
On the other hand, if the set of admissible evolutions for the asset is 
\begin{equation*}
\Lambda = \Big\{(f,t) \in \Scal : \sup_{s \le t} f(s)-c \le
\frac{1}{t} \int_0^tf(s)\, \di s\Big\},
\end{equation*}
which is a drawdown constraint where the constraint depends on the running average of the price process, then we are not able to replace $\Lambda$ by a `nice' set $\Lambda'\subset\Lambda$ as above.
Here the class of admissible stopping times is certainly bounded
above by a stopping time ($\inf\{t \ge 0: \max\{t^{-1}\int_0^{t}f(s)\, 
\di s, f(t)\}\
\le \sup_{s \le t} f(s)-c\}$), but it is easily seen that there are inadmissible paths which occur before this time.
\medskip\\
In what follows we will consider the cases in Theorem~\ref{thm:RSTExists} separately, analysing them in specific settings. In particular, in Sections~\ref{sect:AY} and~\ref{sect:Root} we present two frameworks where the additional information $\Lambda$ is of the kind $\Lambda_1$ in \eqref{ex:l}, and we are able to give necessary and sufficient conditions for the set $\RST(\mu;\Lambda)$ to be non-empty, hence strengthening the result in case~\ref{item:1_RST} of Theorem~\ref{thm:RSTExists}. In Theorem \ref{thm:constRoot}, we will exemplify the power of the monotonicity principle by showing the structure of the solutions to a Root-type optimisation problem with an Az\'ema-Yor-type constraint.
Moreover, in Section~\ref{sect:var} we consider the additional information $\Lambda$ to be of the kind $\Lambda_2$ in \eqref{ex:l} and, for options on variance, we determine the primal optimisers by means of our constrained monotonicity principle (Theorem~\ref{thm:mp}), as well as the dual optimisers.

We remark that the first two cases imply results and constraints for the third case also, e.g.\ Theorems \ref{thm:exay} and \ref{thm:root} directly imply necessary conditions for the third case. More generally, using the monotonicity principle Theorem \ref{thm:mp} one can derive the corresponding versions of Root and Az\'ema-Yor embedding with a general time-space starting law (cf.\ Section \ref{sect:var} for the case of Root). Using similar arguments as in the proof of Theorem \ref{thm:exay} and \ref{thm:root}, with slightly more notation, one can derive the corresponding versions of these results  keeping also track of the condition $\underline\tau \leq t$ implying necessary and sufficient conditions for the case $\Lambda_3$. We omit the details.

\subsection{Information as barrier in a certain phase space.}\label{sec:barrier}
We now consider the case where the additional information is of the kind of $\Lambda_1$ in \eqref{ex:l} and translates in having a barrier in a certain phase space. We will see how in this situation the No-Arbitrage condition (cf.\ Proposition \ref{prop:NA}, Theorem \ref{thm:RSTExists}) imposes an order between such a barrier, and the barrier characterising the unique optimal stopping for the uninformed agent in such a phase space. These results are notable since the ordering of barriers is a much weaker condition than the convex order condition, significantly strengthening the results of Theorem \ref{thm:RSTExists}.

\subsubsection{The Root phase space}\label{sect:Root} 
We recall that the Root solution of the (unconstrained) Skorokhod embedding problem for the distribution $\mu$ is given by
\[
\tau_{Root}(\mu)=\inf\{t\geq 0 : (t,W_t)\in\mathcal{R}\},
\]
where $\mathcal{R}$ is a closed \emph{barrier}, that is, $(t,x) \in \mathcal{R}$ implies $(s,x) \in \mathcal{R}$ for $s>t$; see \cite{Root69}. This is one of the first known solutions to SEP, and is optimal when $\gamma(f,t)=h(t)$ for a strictly convex function $h$.  The Root solution is illustrated in Figure~\ref{fig:Root}.  To avoid trivialities, we assume that our barriers are \emph{regular} (see \cite{COT15}), that is, they are closed and $\{x: (0,x) \not\in \mathcal{R}\}$ is an open interval, containing the origin; any barrier which is not regular can be replaced by a regular barrier without changing the hitting time. Any regular barrier can be described by its \lsc{} barrier function $\mathcal R(x)=\inf\{t:(t,x)\in\mathcal R\}.$
\begin{figure}[ht]
  \centering
\begin{asy}[width=0.7\textwidth]
  import graph;
  import stats;
  import patterns;

  // import settings;

  // gsOptions="-P"; 

  // Construct a Brownian motion of time length T, with N time-steps
  int N = 200;
  //int N = 300;
  real T = 1.6*1.6;
  real dt = T/N;
  real B0 = 0;

  real sig = 0.7;

  real xmax = 0.8;
  real xmin = -0.8;

  real tmax = (xmax-xmin)*1.6;

  real[] B; // Brownian motion
  real[] t; // Time

  path BM;

  // Seed the random number generator. Delete for a "random" path:
  srand(123);

  B[0] = B0;
  t[0] = 0;

  BM = (t[0],B[0]);

  // Define a barrier

  real R(real y) {return 1.85 + ((y-1)**4)/8-(y**2)/1.5 -(y**6)*4;}

  int H = N+1;
  int H2 = N+1;
  int BMstop;
  int BMstop2 = N+2;

  for (int i=1; i<N+1; ++i)
  {
    B[i] = B[i-1] + sig*Gaussrand()*sqrt(dt);
    t[i] = i*dt;

    if ((H==N+1)&&(t[i]>=R(B[i])))
    {
	H = i;
	BMstop = length(BM);
       BM = BM--(R(B[i]),B[i]);
    }
    else
    {
      BM = BM--(t[i],B[i]);
    }

  }

  if (H==N+1)
  BMstop = length(BM);

  pen p = deepgreen + 1.5;
  pen p2 = lightgray + 0.25;
  //pen p2 = mediumgray + 1;

  //if (H<N+1)
  //draw(subpath(BM,BMstop,BMstop2),p2);

  pair tau = point(BM,BMstop+1);

  pen q = black + 0.5;

  real eps1 = 0.05;
  real eps2 = 0.15;

  path barrier = (graph(R,identity,xmin+eps1,xmax-eps1)--(tmax-eps2,xmax-eps1)--(tmax-eps2,xmin+eps1)--cycle);

  add("hatch",hatch(1mm,W,mediumgray));

  fill(barrier,pattern("hatch"));

  draw(graph(R,identity,xmin+eps1,xmax-eps1),NW,deepblue+0.5);

  draw((0,xmin)--(0,xmax+eps1),q,Arrow);
  draw((0,0)--((T+eps1),0),q,Arrow);

  draw((tau.x,0)--tau,mediumgray+dashed);
  label("$\tau_{Root}(\mu)$",(tau.x,0),S);

  draw(subpath(BM,0,BMstop+1),p);

  label("$t$",(T+eps1,0),S);
  label("$W_t$",(0,(xmax+eps1)),(-1,0));

  label("$\mathcal{R}$",(2,0.3),UnFill(0.5mm));

  //label("$D_{\Rt}$",(2.2,1.35),UnFill(0.5mm));
\end{asy}

  \caption{The Root solution to the SEP.}
  \label{fig:Root}
\end{figure}

For the informed agent we assume that
\[
\Lambda=\{r(\omega,t) : t\leq\overline{\tau}\},
\]
where the stopping time $\overline{\tau}$ is the hitting time of a regular barrier $\Bcal$ in the phase space $(t,W)$, i.e. a Root-type barrier:
\begin{equation}\label{eq:sr}
\overline{\tau}=\inf\{t\geq 0 : (t,W_t)\in\Bcal\}.
\end{equation}

As in Theorem~\ref{thm:exay}, we are able to determine whether
$\RST(\mu;\Lambda)$ is empty, and hence whether there is an arbitrage
for the informed agent, through properties of the barriers.

\begin{theorem}\label{thm:root}
  Let the set $\Lambda$ be given by \eqref{eq:1}, with
  $\overline{\tau}$ of the form in \eqref{eq:sr}. Then the set
  $\RST(\mu;\Lambda)$ is non-empty if and only if:
  \begin{equation}\label{eq:hb2}
    \Bcal \subseteq \Rcal,
  \end{equation}
  which yields $\tau_{Root}(\mu)\leq\overline{\tau}$. In particular, if
  \eqref{eq:hb2} holds, the stopping rule
  $\tau_{Root}(\mu)$ is admissible for the informed agent, in the sense that $((W_t)_{t\leq\tau_{Root}(\mu)},\tau_{Root}(\mu))\in\Lambda$ a.s.
\end{theorem}

\begin{proof}
  We first observe that if \eqref{eq:hb2} holds, then we immediately
  have $\tau_{Root}(\mu)\leq\overline{\tau}$, and since $\tau_{Root}(\mu)\in
  \RST(\mu;\Lambda)$, then $\RST(\mu;\Lambda) \neq \emptyset$.  To show the reverse implication, suppose, for contradiction,
  that $\RST(\mu;\Lambda)$ is non-empty and $\Bcal \not\subseteq
  \Rcal$. This means that there exist pairs
  $(t,x)\in\Bcal\setminus\Rcal$. Among those pairs, we consider a
  fixed $(\hat t, \hat x)$ such that there are no
  $(t,\hat x)\in\Bcal\setminus\Rcal$ with $t<\hat t$, as in Figure~\ref{fig:root2}.
  \begin{figure}[ht]
    \centering
    \begin{asy}[width=0.7\textwidth]
      import graph;
      import stats;
      import patterns;

      // import settings;

      // gsOptions="-P"; 

      // Construct a Brownian motion of time length T, with N time-steps
      int N = 200;
      //int N = 300;
      real T = 1.6*1.6;
      real dt = T/N;
      real B0 = 0;

      real sig = 0.7;

      real xmax = 0.8;
      real xmin = -0.8;

      real tmax = (xmax-xmin)*1.6;

      real[] B; // Brownian motion
      real[] t; // Time

      path BM;

      // Seed the random number generator. Delete for a "random" path:
      // srand(130);
      srand(137);

      B[0] = B0;
      t[0] = 0;

      BM = (t[0],B[0]);

      // Define a barrier

      real R(real y) {return 1.85 + ((y-1)**4)/8-(y**2)/1.5 -(y**6)*4;}
      real R2(real y) {return 1.55 + ((y+1)**4)/6+(y**2)/1.5 -(y**7)*4-(max(y,0)**4)*5;}

      int H = N+1;
      int H2 = N+1;
      int BMstop;
      int BMstop2 = N+2;

      for (int i=1; i<N+1; ++i)
      {
        B[i] = B[i-1] + sig*Gaussrand()*sqrt(dt);
        t[i] = i*dt;

        if ((H==N+1)&&(t[i]>=R2(B[i])))
        {
          H = i;
          BMstop = length(BM);
          BM = BM--(R2(B[i]),B[i]);
        }
        else
        {
          BM = BM--(t[i],B[i]);
        }

      }

      if (H==N+1)
      BMstop = length(BM);

      pen p = palegreen + 1.5;
      pen p2 = lightgray + 0.25;
      //pen p2 = mediumgray + 1;

      //if (H<N+1)
      //draw(subpath(BM,BMstop,BMstop2),p2);

      pair tau = point(BM,BMstop+1);

      pen q = black + 0.5;

      real eps1 = 0.05;
      real eps2 = 0.15;

      // path barrier = (graph(R,identity,xmin+eps1,xmax-eps1)--(tmax-eps2,xmax-eps1)--(tmax-eps2,xmin+eps1)--cycle);

      // add("hatch",hatch(1mm,W,mediumgray));

      // fill(barrier,pattern("hatch"));

      draw((0,xmin)--(0,xmax+eps1),q,Arrow);
      draw((0,0)--((T+eps1),0),q,Arrow);

      // draw((tau.x,0)--tau,mediumgray+dashed);
      // label("$\tau_{Root}(\mu)$",(tau.x,0),S);

      draw(subpath(BM,0,BMstop+1),p);

      draw(graph(R,identity,xmin+eps1,xmax-eps1),NW,deepblue+1.5);
      draw(graph(R2,identity,xmin+eps1,xmax-eps1),NW,heavymagenta+1.5);

      label("$t$",(T+eps1,0),S);
      label("$W_t$",(0,(xmax+eps1)),(-1,0));

      label("$\mathcal{R}$",(R(0.4),0.45),E,deepblue);
      label("$\mathcal{B}$",(R2(-0.5)+0.05,-0.5),E,heavymagenta);
      
      label("$(\hat{t},\hat{x})$",tau,E);
      dot(tau);

      //label("$D_{\Rt}$",(2.2,1.35),UnFill(0.5mm));
    \end{asy}

    \caption{Proof of Theorem~\ref{thm:root}.}
    \label{fig:root2}
  \end{figure}

Now consider $\tau'\in\RST(\mu;\Lambda)$. Denote the local time of Brownian motion in $z$ by $L^z$. Since the Root embedding maximises $\E\left[L^x_{\tau\wedge t}\right]$ among all stopping times $\tau$ which are minimal embeddings of $\mu$ (cf.\ \eqref{eq:minimal}), simultaneously for all
$(t,x)\in\R_+\times\R$ (\eg{} by \cite[Theorem~3]{GaObRe15}), then in  particular
\[
\E\left[L^{\hat x}_{\tau'\wedge \hat t}\right]\leq \E\left[L^{\hat x}_{\tau_{Root}(\mu)\wedge \hat t}\right].
\]
On the other hand, the path stopped at $\tau'$ cannot accumulate any
more local time at $\hat x$ after $\hat t$, \ie{} $\E\left[L^{\hat
    x}_{\tau'\wedge \hat t}\right] = \E\left[L^{\hat x}_{\tau'\wedge
    t}\right]$ for all $t \ge \hat{t}$, while the Root stopping rule
will do so ($\E\left[L^{\hat x}_{\tau_{Root}(\mu)\wedge \hat t}\right] <
\E\left[L^{\hat x}_{\tau_{Root}(\mu)\wedge t}\right]$ when $t>\hat{t}$), because the barrier is assumed to be regular. Therefore,
\[
\E\left[L^{\hat x}_{\tau'}\right]=\E\left[L^{\hat x}_{\tau'\wedge \hat t}\right]\leq \E\left[L^{\hat x}_{\tau_{Root}(\mu)\wedge \hat t}\right]<\E\left[L^{\hat x}_{\tau_{Root}(\mu)}\right].
\]
This gives the desired contradiction, since, for any $x\in\R$ and any stopping time $\tau\in\RST(\mu)$,
\[
\E\left[L^x_\tau\right]=\E\left[|W_\tau-x|\right]-|x|=-u_\mu(x)-|x|,
\]
where $u_\mu$ is the potential function associated to $\mu$, i.e., $u_\mu(x)=-\int|y-x|\mu(dy)$.
\end{proof}
\subsubsection{The Az\'ema-Yor phase space}\label{sect:AY}
We let $\ol{\omega}_t:=\sup_{0\leq s\leq t} \omega_s$, and define the process $\ol{W}$ analogously. We start by recalling the Az\'ema-Yor solution of the  (unconstrained) Skorokhod embedding problem (SEP). The barycenter function $b_\mu$ of a probability measure $\mu$ is defined by
\begin{equation*}
b_{\mu}(x) := \frac{\int_{[x,\infty)} y \, \mu(\di y)}{\mu([x,\infty))}.
\end{equation*}
Denote the inverse of $b_\mu$ by $\beta_\mu$. The solution to the SEP by Az\'ema and Yor, see \cite{AY79}, is given by
\[
\tau_{AY}(\mu)=\inf\{t\geq 0 : W_t\leq \beta_\mu(\ol{W}_t)\}.
\]
%where $\beta$ is the  right-continuous inverse of the barycenter function $b_\mu$ associated to
%$\mu$; see \cite{AY79}.
This is arguably the most renowned solution to SEP, for which many properties are known, among which, that it maximises stochastically the maximum of the stopped Brownian motion. See the survey article of \citet{Ob04} for further details.
We illustrate the Az\'ema-Yor solution in Figure~\ref{fig:AY}.
\begin{figure}
\centering
  \begin{asy}[width=0.6\textwidth]
    import graph;
    import stats;
    import patterns;

    // Construct a Brownian motion of time length T, with N time-steps
    int N = 3000;
    real T = 1.25;
    real dt = T/N;
    real B0 = 0;

    real[] B; // Brownian motion
    real[] t; // Time
    real[] M; // Maximum

    path BM;
    path BMM;

    real xmax = 1.0;
    real xmin = -0.15;
    real ymax = xmax;

    // Seed the random number generator. Delete for a "random" path:
    srand(94);

    B[0] = B0;
    t[0] = 0;
    M[0] = B[0];
    
    BM = (t[0],B[0]);
    BMM = (M[0],B[0]);

    //real psiinv(real y) {return y-((y-2)**2)*exp(-y/4)/4-1/2;}
    //real psiinv(real y) {return max(xmin+0.005+exp(-0.5/y)*(-exp(-(y-0.7)**3)+y*1.1+0.5-xmin),xmin);}
    //real psiinv(real y) {return max(xmin+0.005+exp(-0.5/y)*(-exp(-(y-0.7)**3)+y*1.1+0.5-xmin),xmin);}
    real psiinv(real y) {return xmin + 2*y/3 + sin(-8*y)/16;}

    real x0 = 0.01;
    real x1 = xmax-0.025;

    real eps3 = 0.275;
    
    int H = N+1;
    int BMMstop;
    
    for (int i=1; i<N+1; ++i)
    {
      B[i] = B[i-1] + Gaussrand()*sqrt(dt);
      t[i] = i*dt;
      M[i] = (B[i] > M[i-1])? B[i] : M[i-1];
      BM = BM--(t[i],B[i]);
      if (M[i-1] < M[i])
      BMM = BMM--(M[i-1],M[i-1])--(M[i],B[i]);
      else
      {
        if ((H==N+1)&&(B[i]<psiinv(M[i])||(M[i] >= ymax-eps3)))
        {
          H = i;
          BMMstop = length(BMM)+1;
          BMM = BMM--(M[i],psiinv(M[i]));
        }
        else
        {
          BMM = BMM--(M[i],B[i]);
        }
      }
    }
    
    if (H==N+1)
    BMMstop = length(BMM);

    pen p = deepgreen + 1.5;
    pen p2 = lightgray + 0.25;
    
    //if (H<N+1)
    //draw(subpath(BMM,BMMstop,length(BMM)),p2);

    pair tau = point(BMM,BMMstop);

    //draw((B0,psiinv(tau.y))--(tau.y,psiinv(tau.y)),p2+dashed);
    
    pen q = black + 0.5;

    real eps1 = 0.05;

    real eps2 = 0.15;

    path barrier = (graph(identity,psiinv,x0,x1)--(max(ymax-eps2,x1),psiinv(x1))--(max(ymax-eps2,x1),psiinv(x0))--cycle);

    add("hatch",hatch(1mm,W,mediumgray));
    
    fill(barrier,pattern("hatch"));

    //Label L = Label("$\Psi_{\mu}(B_t^0)$",UnFill);

    draw(graph(identity,psiinv,x0,x1),NW,deepblue+0.5);

    draw((B0,B0)--(xmax-eps1,xmax-eps1));

    draw(subpath(BMM,1,BMMstop),p);
    draw((0,xmin)--(0,xmax),q,Arrow);
    draw((0,0)--((ymax+eps1),0),q,Arrow);
    label("$W_t$",(0,xmax),(-1,0));
    //label("$\sup_{s \le t} W_t$",(0,(xmax+eps1)),(0,1));
    label("$\ol{W}_t$",((ymax+eps1),0),S);
    
    draw(Label("$\tau_{AY}$",UnFill(0.25mm)),(tau.x+0.04,tau.y-0.04));   
  \end{asy}
  \caption{The Az\'ema-Yor construction.}
\label{fig:AY}
\end{figure}

For the informed agent, we assume that
\begin{equation}\label{eq:1}
\Lambda=\{r(\omega,t) : t\leq\overline{\tau}\},
\end{equation}
where the stopping time $\overline{\tau}$ is the hitting time of a barrier in the phase space $(\ol{W},W)$:
\begin{equation*}
  \overline{\tau}=\inf\{t\geq 0 : (\ol{W}_t,W_t)\in\Hcal\},
\end{equation*}
where $\Hcal$ is a Borel set $\Hcal\subseteq \{(x,y)\in\R_+\times\R : y\leq x\}$ induced by some increasing left-continuous Borel function $h:\R_+\to\R$ via
$$\Hcal = \{(x,y) :  y \leq h(x) \},$$
so that $(x,y)\in\Hcal$ and $z>x$ imply $(z,y)\in\Hcal$.
Note that this  gives
\begin{equation}\label{eq:say}
\overline{\tau}=\inf\{t\geq 0 : W_t\leq h(\ol{W}_t)\},
\end{equation}
 thus the set $\Lambda$ in \eqref{eq:1}
corresponds to the following set of feasible paths for the informed agent:
\begin{equation}\label{eq:a1}
\mathcal{A}=\{\omega\in\Ccal_{s_0}[0,T] : \omega_t>h(\ol{\omega}_t)\; \forall t\in[0,T)\},
\end{equation}
that is, the paths that satisfy the drawdown constraint  $\omega>h(\ol{\omega})$ during the period $[0,T)$.

We now give a result which shows that, when the agent's information is
given by $\mathcal{A}$ as in \eqref{eq:a1}, then we can provide a simple necessary and sufficient
condition for the existence of consistent models for the informed agent, cf. \ref{item:1_RST} of
Theorem~\ref{thm:RSTExists}. If there is no ambiguity we write $\beta_\mu=\beta$ in the following.

\begin{theorem}\label{thm:exay}
  Let the set $\Lambda$ be given by \eqref{eq:1}, with
  $\overline{\tau}$ of the form in \eqref{eq:say}. Then the set
  $\RST(\mu;\Lambda)$ is non-empty if and only if:
  \begin{equation}\label{eq:hb}
    h(x)\leq\beta(x) \quad \textrm{for all $x\in\R_+$},
  \end{equation}
  which yields $\tau_{AY}(\mu)\leq\overline{\tau}$. In particular, if \eqref{eq:hb}
  holds, the stopping rule $\tau_{AY}(\mu)$ is admissible for the
  informed agent, in the sense that $((W_t)_{t\leq\tau_{AY}(\mu)},\tau_{AY}(\mu))\in\Lambda$ a.s.
\end{theorem}

\begin{proof} 
  We first observe that if \eqref{eq:hb} holds, then we immediately
  have $\tau_{AY}(\mu) \le \overline{\tau}$, and since $\tau_{AY}(\mu) \in
  \RST(\mu;\Lambda)$, then $\RST(\mu;\Lambda) \neq \emptyset$.

  For the reverse implication, we suppose that there exists $\hat
  x\in\R_+$ such that $h(\hat x)>\beta(\hat x)$, as in Figure~\ref{fig:pfAY}.
\begin{figure}
\centering
  \begin{asy}[width=0.6\textwidth]
    import graph;
    import stats;
    import patterns;

    // Construct a Brownian motion of time length T, with N time-steps
    int N = 3000;
    real T = 1.25;
    real dt = T/N;
    real B0 = 0;

    real[] B; // Brownian motion
    real[] t; // Time
    real[] M; // Maximum

    path BM;
    path BMM;

    real xmax = 1;
    real xmin = -0.15;
    real ymax = xmax;

    // Seed the random number generator. Delete for a "random" path:
    srand(94);

    B[0] = B0;
    t[0] = 0;
    M[0] = B[0];
    
    BM = (t[0],B[0]);
    BMM = (M[0],B[0]);

    //real psiinv(real y) {return y-((y-2)**2)*exp(-y/4)/4-1/2;}
    //real psiinv(real y) {return max(xmin+0.005+exp(-0.5/y)*(-exp(-(y-0.7)**3)+y*1.1+0.5-xmin),xmin);}
    //real psiinv(real y) {return max(xmin+0.005+exp(-0.5/y)*(-exp(-(y-0.7)**3)+y*1.1+0.5-xmin),xmin);}
    real psiinv(real y) {return xmin + 2*y/3 + sin(-8*y)/16;}
    real psiinv2(real y) {return xmin + 3*y/4 + cos(-8*y)/16-1/20;}

    real x0 = 0.0;
    real x1 = xmax-0.025;

    real eps3 = 0.275;
    
    int H = N+1;
    int BMMstop;
    
    for (int i=1; i<N+1; ++i)
    {
      B[i] = B[i-1] + Gaussrand()*sqrt(dt);
      t[i] = i*dt;
      M[i] = (B[i] > M[i-1])? B[i] : M[i-1];
      BM = BM--(t[i],B[i]);
      if (M[i-1] < M[i])
      BMM = BMM--(M[i-1],M[i-1])--(M[i],B[i]);
      else
      {
        if ((H==N+1)&&(B[i]<max(psiinv(M[i]),psiinv2(M[i]))||(M[i] >= ymax-eps3)))
        {
          H = i;
          BMMstop = length(BMM)+1;
          BMM = BMM--(M[i],max(psiinv(M[i]),psiinv2(M[i])));
        }
        else
        {
          BMM = BMM--(M[i],B[i]);
        }
      }
    }
    
    if (H==N+1)
    BMMstop = length(BMM);

    pen p = palegreen + 0.5;
    pen p2 = lightgray + 0.25;
    
    //if (H<N+1)
    //draw(subpath(BMM,BMMstop,length(BMM)),p2);

    pair tau = point(BMM,BMMstop);

    //draw((B0,psiinv(tau.y))--(tau.y,psiinv(tau.y)),p2+dashed);
    
    pen q = black + 0.5;

    real eps1 = 0.05;

    real eps2 = 0.15;

    //path barrier = (graph(identity,psiinv,x0,x1)--(max(ymax-eps2,x1),psiinv(x1))--(max(ymax-eps2,x1),psiinv(x0))--cycle);

    //add("hatch",hatch(1mm,W,mediumgray));
    
    //fill(barrier,pattern("hatch"));

    //Label L = Label("$\Psi_{\mu}(B_t^0)$",UnFill);

    draw(graph(identity,psiinv,x0,x1),NW,deepblue+1.0);
    label("$h(\ol{W}_t)$",(x1,psiinv(x1)),E,deepblue+1.0);
    draw(graph(identity,psiinv2,x0,x1),NW,heavymagenta+1.0);
    label("$\beta(\ol{W}_t)$",(x1,psiinv2(x1)),E,heavymagenta+1.0);

    draw((B0,B0)--(xmax-eps1,xmax-eps1));

    real xh = 0.5;
    draw((xh,0)--(xh,psiinv(xh)),black+1.0+dashed);
    label("$\hat{x}$",(xh,0),S);

    draw(subpath(BMM,1,BMMstop),p);
    draw((0,xmin)--(0,xmax),q,Arrow);
    draw((0,0)--((ymax+eps1),0),q,Arrow);
    label("$W_t$",(0,xmax),(-1,0));
    //label("$\sup_{s \le t} B_t$",(0,(xmax+eps1)),(0,1));
    label("$\ol{W}_t$",((ymax+eps1),0),S);
    
    //draw(Label("$\tau_{AY}$",UnFill(0.25mm)),(tau.x+0.04,tau.y-0.04));   
  \end{asy}
  \caption{Proof of Theorem~\ref{thm:exay}}
  \label{fig:pfAY}
\end{figure}
Then we fix $\tau'\in \RST(\mu;\Lambda)$, and argue as follows. Define a measure
\begin{equation*}
  \eta(A):= \P(W_{\tau'} \in A, \ol{W}_{\tau'} \ge \hat{x})
\end{equation*}
and note that, by the martingale property, $\int y \,\eta(\di y) =
\hat{x} \cdot \eta(\R)$. Moreover, $\eta(A \cap [\hat{x},\infty)) = \mu(A
\cap [\hat{x},\infty))$, and $\eta(A) \le \mu(A)$ for all Borel sets
$A$.

Define functions $\Phi_{\eta}, \Phi_\mu:(-\infty,\hat{x}] \to \R$ by:
\begin{equation*}
  \Phi_\eta(x) = \int_{[x,\infty)} y \, \eta(\di y) - \hat{x} \cdot
  \eta([x,\infty)) = \int_{[x,\infty)} (y-\hat{x}) \, \eta(\di y),
\end{equation*}
and similarly for $\mu$. Then $\Phi_\mu, \Phi_\eta$ are both
increasing on $(-\infty, \hat{x}]$, $\Phi_\mu(\hat{x}) =
\Phi_\eta(\hat{x})$, and $\Phi_\mu(x) - \Phi_\eta(x)$ is increasing in
$x$ for $x \in (-\infty,\hat{x}]$ since $\mu(\di y) \ge \eta(\di
y)$. Hence we deduce that $\Phi_\mu(x) \le \Phi_\eta(x)$ for $x \le
\hat{x}$.

Now we observe that 
 $\eta((-\infty,h(\hat{x}))) = 0$, so
$\Phi_\eta(h(\hat{x})) = 0$. On the other hand, by the definition of the
barycentre function,
\begin{equation*}
  \beta(\hat{x}) := \sup\left\{y <\hat{x}: \Phi_\mu(y) \le 0\right\}.
\end{equation*}
It follows from $\Phi_\mu(x) \le \Phi_\eta(x)$ that $h(\hat{x}) \le
\beta(\hat{x})$, contradicting our original assumption.
\end{proof}

Let us consider the drawdown constraint in \eqref{eq:34} which corresponds to $\Lambda=\lc 0,\overline{\tau}\rc$, with $\overline{\tau}$ given as in \eqref{eq:say} for $h(x)=x-c$. In this case Theorem~\ref{thm:exay} implies that $x-c\leq\beta(x)=b_\mu^{-1}(x)$ must hold in order to have a feasible solution for the informed agent. This condition can of course be rephrased in terms of barycentre functions since $h$ is the inverse
of the barycenter function associated to $\overline \mu\sim W_{\overline{\tau}}$ (by  \cite{AY79}).
Therefore, the existence of a consistent solution for the insider is equivalent to $b_{\overline \mu}^{-1}\leq b_\mu^{-1}$, that is,
$b_\mu\leq b_{\overline \mu}$.

Theorem~\ref{thm:exay} tells us when the pricing problem for the insider, cf. \eqref{eq:insider} or \eqref{eq:P*f}, has a feasible solution, and hence an optimiser under the conditions of Theorem~\ref{thm:max exists}, but does not tell us anything about the specific optimiser. On the other hand, the constrained monotonicity principle, Theorem \ref{thm:mp}, allows us to characterise the geometry of optimisers in various settings. We illustrate this in the special situation where the agent wishes to find the stopping times $\tau$ solving \eqref{eq:Wiener opt} in the case $\gamma(f,s)=-s^2$ corresponding to the payoff $F(S)=-\langle S\rangle_T^2.$ We are interested in characterising solutions to \begin{align}\label{eq:constRoot}
 \min_{\tau\in\RST(\mu;\Lambda)} \E[\tau^2],
\end{align}
where
\begin{equation}\label{eq.consh}
\Lambda=\{r(\omega,t):t\leq \overline\tau\},\;\; \text{with}\;\; \overline\tau=\inf\{t\geq 0 : W_t\leq h(\overline W_t)\},
\end{equation}
for a step function $h(x)=\sum_{i=1}^n a_i 1_{[m_{i-1},m_i)}(x)$ with $m_0=s_0 <m_1<\ldots<m_n$ and $a_1\leq a_2\leq\ldots\leq a_n$.
We set $\tilde\Lambda=\{r(\omega,t): t <\overline\tau\}$, and note that
 $(f,s)\in\tilde\Lambda$ and $m_{i-1}\leq \ol{f}_s <m_i$ imply that $f(s)>a_i$.

We recall that in the unconstrained case, i.e.\ $h(x)=-\infty$, the solution is the Root solution $\tau_\mathsf{Root}$, the first hitting time of a barrier in space-time (see also Section \ref{sect:Root}).
% More precisely, there is a set $\mathcal R\subset \R_+\times\R$ with the property that $(s,x)\in\mathcal R$ implies $(t,x)\in\mathcal R$ for all $t>s$ such that
% $$\tau_\mathsf{Root}=\inf{t\geq 0 : (t,B_t)\in\mathcal R\}.$$

In the current setup, the situation is similar:

\begin{theorem}\label{thm:constRoot}
 Assume that $\RST(\mu;\Lambda)\neq \emptyset$ and that the optimisation problem \eqref{eq:constRoot} is well posed. Then for any optimiser $\hat\tau$ there exists a sequence of barriers
  $(\mathcal R_i)_{i=1}^n$ such that 
 $$\hat\tau= \inf\{t\geq 0: (t,W_t)\in \mathcal R_{\ell(\ol{W}_t)}\},$$
 where $\ell(m)=\sum_{i=1}^n i1_{[m_{i-1},m_i)}(m).$ 
 
Moreover, for each $j\leq i$ it holds that
$$ \left(\mathcal R_j\cap [0,\infty)\times (a_i,\infty)\right) \subset \left(\mathcal R_i\cap [0,\infty)\times (a_i,\infty) \right).$$
\end{theorem}

\begin{proof}
  To avoid too many minus signs, we redefine $\gamma(f,s)=s^2$ and for this proof we consider the minimisation variant of \eqref{eq:Wiener opt}.

  By Theorem \ref{thm:max exists} we can find a minimiser, say $\hat\tau$, to the optimisation problem \eqref{eq:constRoot}. By Theorem \ref{thm:mp} we can pick a feasible $\gamma$-monotone set $\Gamma$ such that $\hat\tau(r^{-1}(\Gamma))=1$ and $\SG_\Lambda\cap (\Gamma^<\times \Gamma)=\emptyset$.

%By Remark \ref{rem:SG}, it follows that
 We claim that
\begin{align}\label{eq:SGconstRoot}  
%\SG_\Lambda \supset \{(f,s),(g,t)\in\Lambda\subset \Scal:f(s)=g(t), m_{i-1}\leq \bar f_s, \bar g_t <m_i \text{ some }1\leq i\leq n, s>t \}.
\SG_\Lambda \supset \{((f,s),(g,t))\in\tilde\Lambda\times\Lambda:f(s)=g(t), s>t, \ol{f}_s \geq \ol{g}_t \}.
\end{align}
Indeed, pick $(f,s),(g,t)\in\Lambda$ with $s>t$ and $f(s)=g(t)$.  It holds for any $(k,u)\in \Scal$ by convexity of $s\mapsto s^2$ (and since we are considering minimisation instead of maximisation) that
$$ \gamma(f,s)+\gamma(g\oplus k, t+u)<\gamma(f\oplus k,s+u)+\gamma(g,t),$$
so that \eqref{eq:SGpairs} follows by two observations. First, $(f,s)\in \tilde\Lambda$ implies the existence of at least one Brownian stopping time $\sigma$ with $0<\E[\sigma]<\infty$ such that $1_\Lambda(f\oplus W,s+\sigma)=1$, e.g.\ if $m_{i-1}\leq \ol{f}_s<m_i$ then the first hitting time of Brownian motion of $\{\frac12 (f(s)-a_i), \frac12 (m_i-f(s)\}$ is such a stopping time.
Second, any stopping time $\sigma$ with $1_\Lambda(f\oplus  W,s+\sigma)=1$ necessarily satisfies $1_\Lambda(g\oplus W,t+\sigma)=1$, since $f(s)=g(t), \ol{f}_s\geq \ol{g}_t$, and the function $h$ defining $\Lambda$ is increasing.
 
 %Define $\mu_i \mathsf{Law}(B_\tau |\{ m_{i-1}\leq \overline B_\tau < m_i\}).$ 
% We claim that $\tau$ conditioned on $\{ m_{i-1}\leq \overline B_\tau < m_i\}$ is given by the hitting time of a barrier in space-time. 
 Put $\Gamma_i=\Gamma\cap \{(f,s)\in\Scal: m_{i-1}\leq \ol{f}_s<m_i\}$ and set
 \begin{align*} \mathcal R_i^\mathsf{op}:=\{(s,x): \exists (g,t)\in\Gamma_i, g(t)=x, s>t\},\\
  \mathcal R_i^\mathsf{cl}:=\{(s,x): \exists (g,t)\in\Gamma_i, g(t)=x, s\geq t\}.
 \end{align*}
Pick $(g,t)\in\Gamma_i$. Then we claim that 
$$ \inf\{s\in[0,t]: (s,g(s))\in \mathcal R_i^\mathsf{cl}\}\leq t \leq \inf\{s\in[0,t]: (s,g(s))\in \mathcal R_i^\mathsf{op}\}.$$
Since the first inequality holds by construction, suppose for contradiction that $\inf\{s\in[0,t]: (s,g(s))\in \mathcal R_i^\mathsf{op}\}<t$. In this case, there is $s<t$ such that $(g|_{[0,s]},s)=:(f,s)\in\Gamma_i^<$, $(s,f(s))\in\mathcal R_i^\mathsf{op}$ and since $s<t$ it holds that $f(s)>a_i$. Then there exists $(k,u)\in \Gamma_i$ such that $u<s$ and $k(u)=f(s)>a_i$ so that $(k,u)\in\tilde \Lambda$. However, by \eqref{eq:SGconstRoot}, this means that $((f,s),(k,u)) \in \SG_\Lambda\cap (\Gamma^<\times \Gamma)$, which cannot be the case.

Pick $\omega$ such that $(W(\omega)_{0\leq t\leq\hat\tau(\omega)},\hat\tau(\omega))\in\Gamma_i$ for some $1\leq i\leq n$. It then follows that
\begin{align*}
 \tau_i^{\mathsf{cl}}(\omega):=\inf\{t\geq 0: (t, W_t(\omega))\in \mathcal R_i^\mathsf{cl}\}\leq \hat\tau(\omega) \leq \inf\{t\geq 0: (t, W_t(\omega))\in \mathcal R_i^\mathsf{op}\}=:\tau_i^\mathsf{op}(\omega).
\end{align*}
Then, we can conclude the existence of the barriers $(\mathcal R_i)_{i=1}^n$ by the observation that conditionally on the event $\{m_{i-1}\leq  \overline{W}_{\hat\tau}<m_i \}$ it holds $\tau_i^\mathsf{cl}=\tau_i^\mathsf{op}$ a.s.\ by the strong Markov property and the fact that Brownian motion almost surely immediately returns to its starting point.

To show the final claim, note that \eqref{eq:SGconstRoot} implies that at each $x\notin\{m_1,\ldots,m_n\}$ the condition $(g,t)\in\Gamma\cap\Lambda$ with $g(t)=x$ implies $g(t)\in\tilde\Lambda$. Just as in the first part of the proof, it then follows that there is no $(f,s)\in\Gamma^<$ with $f(s)=x, \ol{f}_s\geq \ol{g}_t$ and $s>t$. This gives the result.
 \end{proof}

\begin{example}
 Consider the case of $\mu=\frac13(\delta_{s_0-1} + \delta_{s_0} + \delta_{s_0+1})$, when $s_0>1$.  Let $h$ be given by the inverse barycentre function of $\mu$, i.e.\ $h=b_\mu^{-1}$ which equals
 $$h(x)= (s_0-1) \cdot 1_{[s_0-1,s_0+1/2)}(x) + s_0 \cdot 1_{[s_0+1/2,s_0+1)} + (s_0+1) \cdot 1_{[s_0+1,\infty)}.$$
 In particular, in the Root-type optimisation problem \eqref{eq:constRoot} constrained by $h$  as in \eqref{eq.consh}, any path which reaches level $s_0+1/2$ will not be stopped at $s_0-1$.
 
 The unconstrained Root solution instead is given by the hitting time of 
 $$\mathcal R=\{(t,s_0-1):t\in [0,\infty)\}\cup \{(t,s_0):t\in [a,\infty)\} \cup \{(t,s_0+1):t\in [0,\infty)\}, $$
 for some $a>0$. In particular, there are paths getting arbitrary close to one but which are stopped at $s_0-1$ so that the constrained Root solution is different from the unconstrained one.
 
 Also note that this is not related to the special case of $\mu$ begin atomic. Indeed, keep the same $h$. Consider $\tilde\mu$ to be the uniform measure on $[s_0-1,s_0+1]$ whose inverse barycentre function is given by $b_{\tilde\mu}^{-1}(m)=2m-s_0-1$ so that $\RST(\tilde\mu; \Lambda)\neq \emptyset$ by Theorem~\ref{thm:exay}. By the same reasoning as before, it is immediate to see that the Root solution is different from the constrained Root solution. 
\end{example}

\begin{remark}
 Let us consider \eqref{eq:constRoot} for the case of a general increasing function $h$ yielding a corresponding set of feasible paths $\Lambda$  as in \eqref{eq.consh}. Assume $\RST(\mu;\Lambda)\neq \emptyset.$ Approximate $h$ from below by step functions $h^n$ with corresponding sets of feasible paths $\Lambda^n$ and the property that $h^n\leq h^{n+1}$. Since then $\Lambda^n \supseteq \Lambda^{n+1}\supseteq\Lambda$, it follows that $\RST(\mu;\Lambda^n)\neq\emptyset$. For each $n$, pick by Theorem \ref{thm:max exists} an optimiser to (the corresponding version of) \eqref{eq:constRoot}, say $\hat\tau^n$. Since $\RST(\mu;\Lambda^n)\supseteq \RST(\mu;\Lambda)$, it follows that for all $n$
 $$\E[(\hat\tau^n)^2]\leq \inf_{\tau\in\RST(\mu;\Lambda)}\E[\tau^2].$$
 Since $\RST(\mu;\Lambda^1)$ is compact and $\hat\tau^n\in\RST(\mu;\Lambda^1)$ for all $n$, there is a converging subsequence and any limit point $\hat\tau$ must lie in $\RST(\mu;\Lambda).$ Moreover, any limit point $\hat\tau$ must be an optimiser by monotonicity, since
 $$\E\left[(\hat\tau^n)^2\right] \leq\E\left[(\hat\tau^{n+1})^2\right]\leq\inf_{\tau\in\RST(\mu;\Lambda)}\E[\tau^2].$$
 Since, by Theorem \ref{thm:constRoot}, each $\hat\tau^n$ is given as the hitting time of  barriers in space-time indexed by the running maximum, it is then plausible to conjecture that this remains true for $\hat\tau$ as well. To make this argument rigorous seems to be outside the scope of this article, however we note that \eqref{eq:SGconstRoot} still holds in the limit.
\end{remark}

\begin{remark}
 Considering in \eqref{eq:constRoot} a maximisation problem instead of a minimisation problem, the corresponding version of  \eqref{eq:SGconstRoot} turns into
 $$ \SG_\Lambda \supset \{((f,s),(g,t))\in\tilde\Lambda\times\Lambda:f(s)=g(t), s<t, \ol{f}_s \leq \ol{g}_t \}.
$$
Following the line of reasoning of Theorem \ref{thm:constRoot}, one can show that the optimal stopping time will be the hitting of a sequence of \emph{inverse barriers}
indexed by the running maximum, i.e.\ the corresponding version of  constrained Rost solutions.
Using similar ideas one can identify the optimal solutions and worst case scenarios in various different setups.
\end{remark}

\subsection{Option pricing in the presence of insider information: Variance options}\label{sect:var}

In this section, we consider the impact on the insider's pricing
bounds which come from additional information. Specifically, we
suppose that the information is on the drawdown, in a similar manner to the previous discussion, for example, as in \eqref{eq.consh}, and we look to find bounds on the prices of options on variance: that is, we consider the motivating example from the introduction, where we think of a trader who believes that the CEO of the company is attempting to satisfy a drawdown constraint, and wishes to understand the impact on pricing bounds of variance options on the same company.

To understand the structure of the derivatives, we consider an asset which follows a model of the form: $\di S_t = S_t \sigma_t \, \di W_t$, where $S_t$ is the discounted asset price, and $W_t$ a Brownian motion. The process $\sigma_t$ is the volatility, and $\int_0^t \sigma^2_r \, \di r$ is known as the \emph{integrated
  variance}. A variance option is then a contract which pays the
holder $G\left(\int_0^t \sigma^2_r \, \di r\right)$. The most common
example is the variance call, where $G(v) = (v-K)_+$. Note that the
integrated variance process can be determined as $\langle \ln
S\rangle_t$, the quadratic variation of the logarithm of the asset
price. For further details, we refer the reader to
\cite{CarrLee:10,Lee:2010ab,CoxWang:11,CoWa13}.

The standard method for pricing such options is to time-change the
process $S_t$ by a time change $\tau_t$ such that $X_t := S_{\tau_t}$
is a \emph{geometric} Brownian motion. With this time change, $(X_t,t)
= (S_{\tau_t},\langle \ln S\rangle_{\tau_t})$, that
is, the time-scale in the transformed picture corresponds to the
integrated variance process. In particular, the problem of finding a
model $S_t$ which minimises $\E[G(\langle\ln S\rangle_T)]$ subject to
$S_T \sim \mu$ is equivalent to finding a stopping time $\tau$ for $X$
to minimise $\E[G(\tau)]$ subject to $X_\tau \sim \mu$.

We would therefore like to compare the minimal (model-independent) price of the variance option for the insider, to that for the uninformed agent. To keep things simple, we consider an option which pays the holder the square-root of the \emph{arithmetic variance}, $V_t := \int_0^t S_r^2 \sigma_r^2 \, \di r$, which corresponds to choosing a time-change $\tau_t$ so that $X_t = S_{\tau_t}$ is a Brownian motion. This places us trivially in the setup of the rest of this paper.

Our problem of interest now may be posed as follows: consider an agent who has inside information on the future evolution of the asset, specifically, who knows that the price will never drop below $h(\ol{S}_t)$, where $h$ is an increasing function.
The agent plans to exploit this information by trading in derivatives written on the asset, and do not have strong modelling beliefs, so wish to profit from their information under any potential model. Suppose variance options with payoff $\sqrt{V_T}$ are liquidly traded. To profit, they plan to sell the derivative and setup a model-independent super-hedging strategy. They want to know at what price-level they are guaranteed to make a profit. If the agent also knows the feasibility set $\Lambda$ given by \eqref{eq.consh}, then their problem becomes to find \begin{align}
  %P^*_\Lambda
  % & =\inf\{\E_\P[\gamma((B_t)_{t\leq\tau},\tau)] : B_\tau\sim\mu,\, 
  %   B_{.\wedge\tau}\ \textrm{is u.i.},\,
  %   \E_\P[1_\Lambda((B_t)_{t\leq\tau},\tau)]=1\} \nonumber\\
  & 
   % = 
    \sup\{\E_\P[\sqrt{\tau}] : W_\tau\sim\mu,\, 
    W_{.\wedge\tau}\ \textrm{is u.i.},\, \tau \le \overline{\tau} \text{ a.s.}
    \}. \label{eq:PRootdefn}
\end{align}
By Theorem~\ref{thm:superhedge}, if we can identify the solution to this
problem, then there exists a corresponding super-hedging
strategy. However, it follows from Theorem~\ref{thm:constRoot} that the solution must be a nested sequence of barriers, which depend on the running maximum. To see how these barriers, and more specifically, the price bound, may depend on the information set, we consider the problem numerically under some additional structural examples.

\subsubsection{Numerical results}\label{sec:numerics}

In this section we illustrate the previous example with some numerical
evidence. In particular, we are interested in illustrating how the
insider's price changes as the information set
changes. 

Our basic setup is as follows: we suppose that the insider's
information set $\Lambda$ is determined by \eqref{eq.consh}, where the function $h$ is of the form: $h(x) = a_1 \1_{[m_0,\infty)}$, that is, there is a single step in the constraint, which comes in at the point where the maximum first exceeds the level $m_0$. In the examples, we will consider the case where the information set changes by varying $m_0$. Moreover, we will assume that the measure to be embedded consists of 4 atoms, at points $\{x_0, x_1, x_2, x_3\}$, and we have $a_1 = x_1$. It follows that the main issue to be determined is the value of the barrier at the level $x_1$ when the level $m_0$ has not yet been reached, and the barrier at $x_2$ both before and after reaching $m_0$. From Theorem~\ref{thm:constRoot}, we know that the barriers at $x_2$ are ordered --- that is, the earliest time at which we stop at $x_2$ before reaching $m_0$, is later than the earliest time we stop at $x_2$ after reaching $m_0$. Since the embedding constraint has two degrees of freedom (there are four atoms of mass, but two values are fixed by the requirement that the probabilities sum to one, and the requirement that the embedded mass has mean equal to $s_0$), this means that we can compute the optimal barrier by optimising over the single remaining degree of freedom.

We implement a simple numerical algorithm, inspired by the PDE characterisation of \cite{CoxWang:11}, which finds the potential of the stopped process. By optimising over the potential functions of the measure embedded before and after reaching $m_0$, we are able to compute the critical times at which the barriers must start. Here, the potential $u_\lambda(x)$ associated with a measure on $\R$ is defined to be: $u_\lambda(x) = \int | y-x| \, \lambda(dy)$. The numerical implementation was performed in Python\footnote{A Jupyter notebook containing the code used to produce the figures in this paper can be downloaded from \url{https://github.com/amgc500/SEPInsider}.}. In Figure~\ref{fig:price} we plot the price of the variance option as a function of $m_0$. Moreover, we can see how the law of the quadratic variance in the extremal model varies as we change $m_0$: this is shown in Figure~\ref{fig:cdf} for several values of $m_0$, as well as the values of the barrier at $x_1, x_2$, before and after hitting $m_0$.

\begin{figure}[ht]
  \centering
  \includegraphics[width=0.9\textwidth]{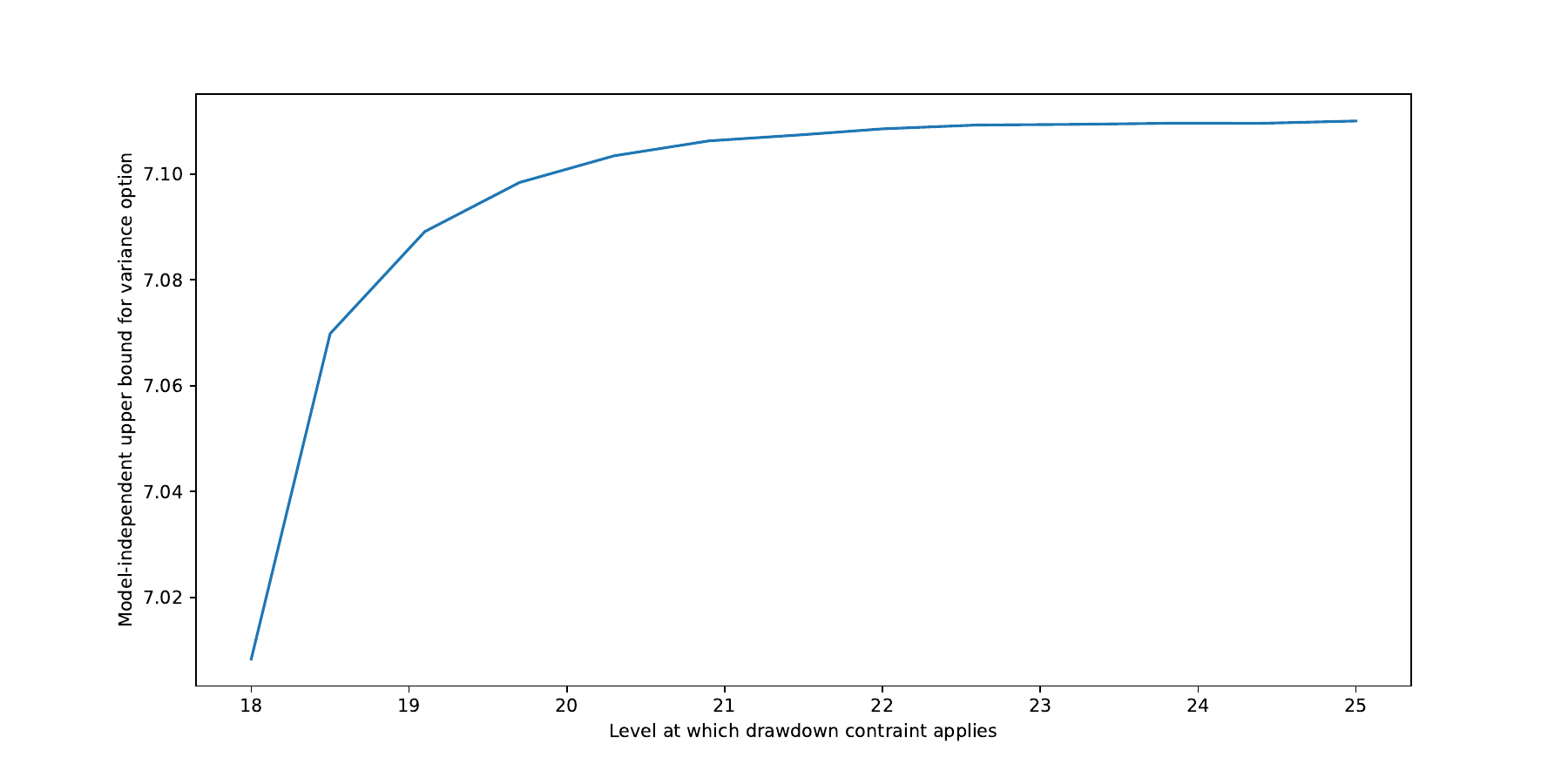}
  \caption{Effect of changing $m_0$ on barrier. For values below about 17.9, there is no feasible model, and therefore there is an arbitrage if the option can be sold at any finite price.  Since the upper range of the price distribution, $x_3 = 25$, $m_0 = 25$ corresponds to having no additional information, the maximal price in the plot (attained when $m_0 = 25$) is accordingly the model-independent upper bound on the price in the absence of any additional information.}
  \label{fig:price}
\end{figure}

\begin{figure}[ht]
  \centering
  \includegraphics[width=0.9\textwidth]{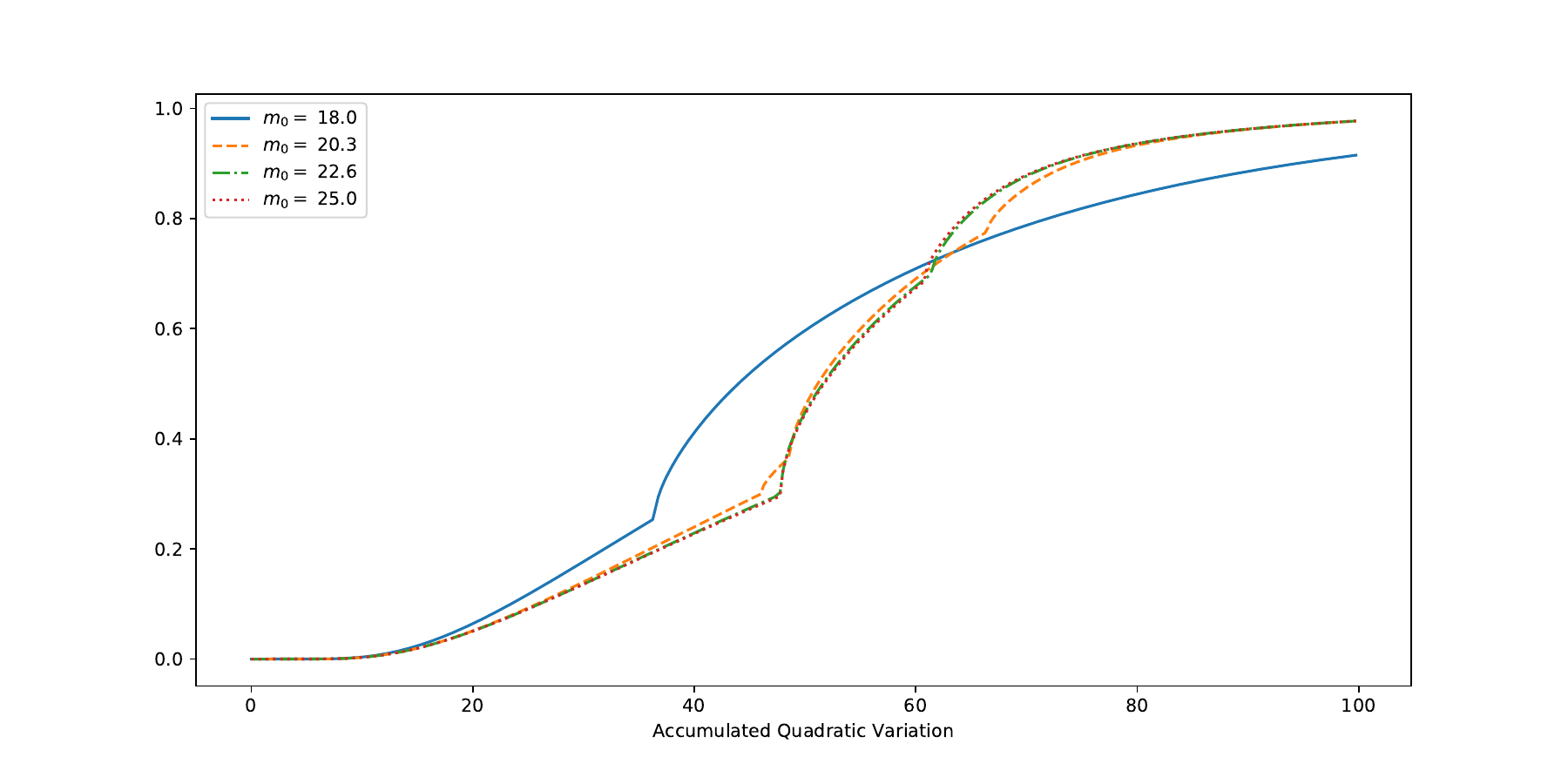}
  \caption{The CDF of the realised quadratic variation of $S$ in the extremal model, for 4 different values of $m_0$.}
  \label{fig:cdf}
\end{figure}

\begin{figure}[ht]
  \centering
  \includegraphics[width=0.9\textwidth]{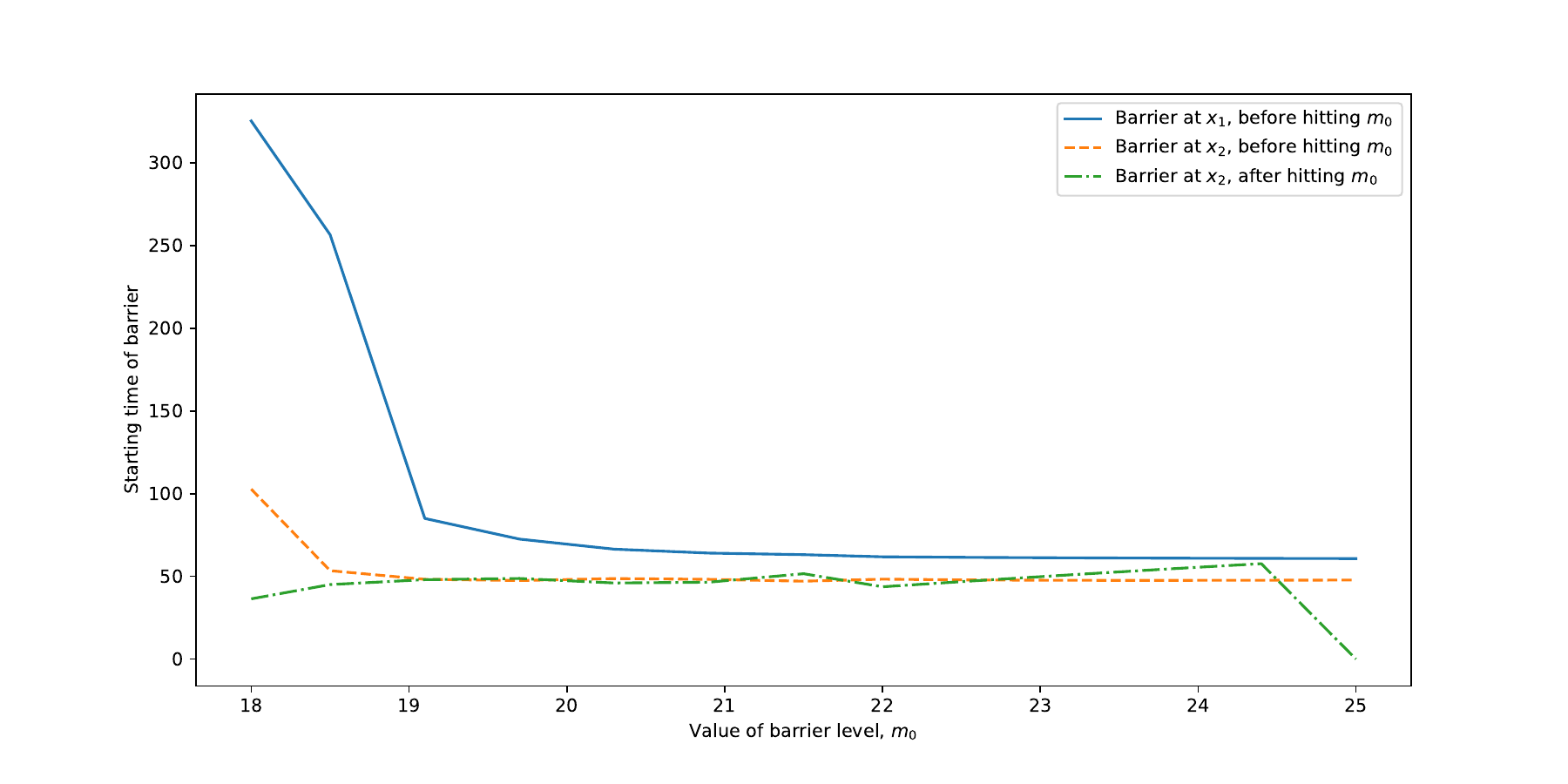}
  \caption{The values determining the barrier for different values of $m_0$. In the figure, the values of the barriers at $x_1, x_2$ before hitting $m_0$, and $x_2$ after hitting $m_0$ are shown. All other barrier values are either $0$ or $\infty$.}
  \label{fig:barrier}
\end{figure}

Under the restriction to a small number of atomic masses, the optimal models are relatively easy to find numerically in simple examples such as these. However, Theorem~\ref{thm:constRoot} only provides necessary conditions for a given barrier to be optimal. An open, and interesting question, is whether it is possible to provide sufficient conditions, and moreover, whether a numerical scheme to compute the corresponding bounds can be implemented. Doing this appears to us to be a challenging problem, and we leave this as an open question for future work.

\renewbibmacro*{in:}{}

\printbibliography

\end{document}

%%% Local Variables:
%%% mode: latex
%%% TeX-master: t
%%% End: